\documentclass[11pt]{article}%

\ifx\SODAVer\undefined%
\newcommand{\SODA}[1]{}
\newcommand{\notSODA}[1]{#1}
\else
\newcommand{\SODA}[1]{#1}
\newcommand{\notSODA}[1]{}
\fi

\makeatletter
\def\input@path{{styles/}}
\makeatother

\SODA{
   \def\CiteCrazy{1}%
}

\SODA{%
   \usepackage[letterpaper, margin=1in]{geometry}

   \usepackage{setspace}
   \singlespacing
   \setlength{\parskip}{0.5em} %
}

\def\UseBibLatex{1}

\providecommand{\BibLatexMode}[1]{}
\providecommand{\BibTexMode}[1]{}

\renewcommand{\BibLatexMode}[1]{#1}
\renewcommand{\BibTexMode}[1]{}

\ifx\UseBibLatex\undefined%
  \renewcommand{\BibLatexMode}[1]{}
  \renewcommand{\BibTexMode}[1]{#1}
\fi

\BibLatexMode{%
    \usepackage[bibencoding=utf8,style=alphabetic,backend=biber]{biblatex}%
    \usepackage{sariel_biblatex}%
    \usepackage{biblatex_authors}
}

\usepackage{amsmath}%
\usepackage{amssymb}%

\usepackage{lipsum} %

\usepackage[table,rgb]{xcolor}%

\definecolor{myDarkRed}{rgb}{0.5, 0.0, 0.0}
\definecolor{myUrlRed}{rgb}{0.25, 0.0, 0.0}
\definecolor{myCiteBlue}{rgb}{0.0, 0.2, 0.445}
\definecolor{myFileBlue}{rgb}{0.0, 0.0, 0.4}
\definecolor{myAnchorBlue}{rgb}{0.0, 0.1, 0.2}

\usepackage[amsmath,thmmarks]{ntheorem}%

\usepackage{microtype}
\usepackage{mleftright}%
\usepackage{xspace}%
\usepackage{graphicx}
\usepackage{hyperref}%
\usepackage[inline]{enumitem} %
\usepackage[ocgcolorlinks]{ocgx2} %
\usepackage{caption}
\usepackage[all]{hypcap} %
\usepackage{hypcap}

\usepackage{tikz}
\hypersetup{%
      unicode,
      breaklinks,%
      colorlinks=true,%
      urlcolor=myUrlRed,%
      linkcolor=myDarkRed,%
      citecolor=myCiteBlue,%
      filecolor=myFileBlue,
      anchorcolor=myAnchorBlue%
}

\usepackage{titlesec}%
\titlelabel{\thetitle. }%

\titlespacing*{\subparagraph}
{0pt}                %
{5pt}                %
{1em}                %

\titleformat{\subparagraph}[runin]
  {\normalfont\scshape} %
  {\thesubparagraph}    %
  {1em}                 %
  {}                    %

\usepackage[amsmath,thmmarks]{ntheorem}%
\usepackage{pgffor} %

\theoremseparator{.}%

\theoremstyle{plain}%
\theoremheaderfont{\normalfont\bfseries}
\theorembodyfont{\itshape}
\newtheorem{theorem}{Theorem}[section]

\theoremseparator{}%

\newtheorem{theoremnp}[theorem]{Theorem}

\theoremseparator{.}%

\foreach \env/\display in {
   lemma/Lemma, conjecture/Conjecture, corollary/Corollary, claim/Claim,  invariant/Invariant,
   proposition/Proposition, prop/Proposition, openproblem/Open Problem%
}{ \edef\temp{{\noexpand\newtheorem{\env}[theorem]{\display}}} \temp }

\theoremseparator{}%

\theoremseparator{.}%

\definecolor{darkslate}{rgb}{0.2, 0.2, 0.2} %

\theoremseparator{.}%
\theoremheaderfont{\sffamily\bfseries\color{darkslate}}
\theorembodyfont{\upshape}%
\foreach \env/\display in {%
   defn/Definition, definition/Definition, example/Example, exercise/Exercise, xca/Exercise, problem/Problem,
   question/Question,assumption/Assumption%
}{\edef\temp{{\noexpand\newtheorem{\env}[theorem]{\display}}} \temp}

\theoremstyle{plain}%
\theoremseparator{.}%
\theoremheaderfont{\sffamily}%
\theorembodyfont{\upshape}%
\newtheorem*{remark:unnumbered}[theorem]{Remark}%

\foreach \env/\display in {
   remark/Remark, note/Note, fact/Fact, observation/Observation
}{ \edef\temp{{\noexpand\newtheorem{\env}[theorem]{\display}}} \temp }

\newcommand{\myqedsymbol}{\rule{2mm}{2mm}}

\theoremheaderfont{\itshape}%
\theorembodyfont{\upshape}%
\theoremstyle{nonumberplain}%
\theoremseparator{}%
\theoremsymbol{\myqedsymbol}%
\newtheorem{proof}{Proof:}%

\definecolor{blue25emph}{rgb}{0, 0, 11}
\SODA{%
\definecolor{almostblack}{rgb}{0.0, 0.2, 0.5} %
}
\notSODA{%
\definecolor{almostblack}{rgb}{0, 0.0, 0.3}%
}

\providecommand{\emphi}[1]{}
\renewcommand{\emphi}[1]{\emphw{#1}}

\providecommand{\emphw}[1]{}%
\renewcommand{\emphw}[1]{{\textcolor{almostblack}{\emph{#1}}}}%
\SODA{%
\renewcommand{\emphw}[1]{{\textcolor{almostblack}{\textbf{#1}}}}%
}

\providecommand{\emphOnly}[1]{}%
\renewcommand{\emphOnly}[1]{\emph{\textcolor{blue25emph}{\textbf{#1}}}}

\newcommand{\SarielThanks}[1]{%
   \thanks{%
      School of Computing and Data Science; %
      University of Illinois; %
      201 N. Goodwin Avenue; %
      Urbana, IL, 61801, USA; %
      \href{mailto:spam@illinois.edu}{sariel@illinois.edu}; %
      \url{http://sarielhp.org/}.%
   #1%
   }%
}

\newcommand{\FaroukThanks}[1]{%
   \thanks{%
      School of Computing and Data Science; %
      University of Illinois; %
      201 N. Goodwin Avenue; %
      Urbana, IL, 61801, USA; %
      \href{mailto:eyfmharb@gmail.com}{eyfmharb@gmail.com}; %
      \url{eyfmharb@gmail.com}.%
   #1%
   }%
}

\newcommand{\QizhengThanks}[1]{%
   \thanks{%
      \href{mailto:hqztrue@sina.com}{hqztrue@sina.com}. %
   #1%
   }%
}

\newcommand{\HLink}[2]{\hyperref[#2]{#1~\ref*{#2}}}
\newcommand{\HLinkY}[2]{\hyperref[#2]{#1}}
\newcommand{\HLinkSuffix}[3]{\hyperref[#2]{#1\ref*{#2}{#3}}}

\newcommand{\figlab}[1]{\label{fig:#1}}
\newcommand{\figref}[1]{\HLink{Figure}{fig:#1}}

\newcommand{\thmlab}[1]{{\label{theo:#1}}}
\newcommand{\thmref}[1]{\HLink{Theorem}{theo:#1}}

\newcommand{\thmrefY}[2]{\HLinkY{#2}{theo:#1}}

\newcommand{\clmlab}[1]{\label{claim:#1}}
\newcommand{\clmref}[1]{\HLink{Claim}{claim:#1}}

\newcommand{\ts}{\hspace{0.6pt}}

\newcommand{\corlab}[1]{\label{cor:#1}}
\newcommand{\corref}[1]{\HLink{Corollary}{cor:#1}}%

\newcommand{\proplab}[1]{\label{prop:#1}}
\newcommand{\propref}[1]{\HLink{Proposition}{prop:#1}}%

\newcommand{\obslab}[1]{\label{observation:#1}}
\newcommand{\obsref}[1]{\HLink{Observation}{observation:#1}}

\newcommand{\tbllab}[1]{\label{table:#1}}
\newcommand{\tblref}[1]{\HLink{Table}{table:#1}}

\newcommand{\factlab}[1]{\label{fact:#1}}%
\newcommand{\factref}[1]{\HLink{Fact}{fact:#1}}%

\newcommand{\lemlab}[1]{\label{lemma_#1}}
\newcommand{\lemref}[1]{\HLink{Lemma}{lemma_#1}}%
\newcommand{\lemrefY}[2]{\hyperref[lemma_#1]{#2}}

\newcommand{\tablab}[1]{\label{table:#1}}%
\newcommand{\tabref}[1]{\HLink{Table}{table:#1}}%
\newcommand{\tabrefY}[2]{\hyperref[table:#1]{#2}}

\newcommand{\seclab}[1]{\label{sec:#1}}
\newcommand{\secref}[1]{\HLink{Section}{sec:#1}}
\newcommand{\secrefY}[2]{\hyperref[sec:#1]{#2}}

\providecommand{\deflab}[1]{\label{def:#1}}
\newcommand{\defref}[1]{\HLink{Definition}{def:#1}}
\newcommand{\defrefY}[2]{\hyperref[def:#1]{#2}}

\providecommand{\eqlab}[1]{}%
\renewcommand{\eqlab}[1]{\label{equation:#1}}

\providecommand{\remove}[1]{}%
\newcommand{\Set}[2]{\left\{ #1 \;\middle\vert\; #2 \right\}}

\newcommand{\pth}[1]{\mleft(#1\mright)}%

\newcommand{\ceil}[1]{\mleft\lceil {#1} \mright\rceil}
\newcommand{\floor}[1]{\mleft\lfloor {#1} \mright\rfloor}

\newcommand{\cardin}[1]{\left\lvert {#1} \right\rvert}%

\renewcommand{\th}{th\xspace}

\renewcommand{\Re}{\mathbb{R}}%

\newlist{compactenumA}{enumerate}{5}%
\setlist[compactenumA]{itemsep=-0.5ex,topsep=0.5ex,partopsep=1ex,parsep=1ex,%
   label=(\Alph*)}%

\newlist{compactenuma}{enumerate}{5}%
\setlist[compactenuma]{itemsep=-0.5ex,topsep=0.5ex,partopsep=1ex,parsep=1ex,%
   label=(\alph*)}%

\newlist{compactenumI}{enumerate}{5}%
\setlist[compactenumI]{itemsep=-0.5ex,topsep=0.5ex,partopsep=1ex,parsep=1ex,%
   label=(\Roman*)}%

\newlist{compactenumi}{enumerate}{5}%
\setlist[compactenumi]{itemsep=-0.5ex,topsep=0.5ex,partopsep=1ex,parsep=1ex,%
   label=(\roman*)}%

\newlist{compactitem}{itemize}{5}%
\setlist[compactitem]{itemsep=-0.5ex,topsep=0.5ex,partopsep=1ex,parsep=1ex,%
   label=\ensuremath{\bullet}}%

\newcommand{\etal}{\textit{et~al.}\xspace}

\usepackage{float}
\usepackage[section]{placeins}

\numberwithin{figure}{section}%
\numberwithin{table}{section}%
\numberwithin{equation}{section}%

\newcommand{\hcite}[2][]{\emph{\textbf{\cite[#1]{#2}.}}}

\newcommand{\NN}{\mathbb{N}}
\newcommand{\ZZ}{\mathbb{Z}}
\renewcommand{\Re}{\mathbb{R}}

\newcommand{\LNK}[2]{%
   \hypersetup{linkcolor=black}{\hyperref[#1]{#2}}
}

\newcommand{\notlab}[1]{\phantomsection\label{notation:#1}}

\newcommand{\gmarkdef}{{g}}
\newcommand{\gmark}{\LNK{notation:g_mark}{\gmarkdef}}

\newcommand{\epidef}{\mathsf{c}^\star}
\newcommand{\epi}{\LNK{def:epicenter}{\epidef}}

\newcommand{\Sin}{s_{\mathrm{in}}}
\newcommand{\Sout}{s_{\mathrm{out}}}

\newcommand{\rindef}{{r}_{\mathrm{in}}}
\newcommand{\rin}{\LNK{notation:r:in}{\rindef}}

\newcommand{\Routdef}{{R}_{\mathrm{out}}}
\newcommand{\Rout}{\LNK{notation:r:out}{\Routdef}}

\newcommand{\nvdef}{{v}}
\newcommand{\nv}{\LNK{notation:v}{\nvdef}}

\newcommand{\bbdef}{{b}}
\newcommand{\bb}{\LNK{notation:b}{\bbdef}}

\newcommand{\iidef}{{i}}
\newcommand{\ii}{\LNK{notation:i}{\iidef}}

\newcommand{\NTX}[1]{\hypersetup{linkcolor=black}{\tabrefY{notations}{#1}}}

\newcommand{\Loptdef}{{L}}
\newcommand{\Lopt}{\LNK{notation:L:opt}{\Loptdef}}

\newcommand{\PoptX}[1]{\Popt_{#1}}
\newcommand{\LoptX}[1]{\Lopt_{#1}}
    \newcommand{\areaOpt}{{\mathsf{A}}}
    \newcommand{\areaSym}{A}
    \newcommand{\perimSym}{L}

\newcommand{\lOpt}{\Lopt}
\newcommand{\lOptX}[1]{\LoptX{#1}}
\newcommand{\WD}{\gmark}

\newcommand{\areaC}{\areaSym}
\newcommand{\perimC}{\perimSym}

\newcommand{\lenX}[1]{\left\| #1 \right\|}

\DeclareMathOperator{\perim}{perim}

\newcommand{\perimX}[1]{\cardin{\partial #1}}

\providecommand{\Barany}{B{\'a}r{\'a}ny\xspace}
\providecommand{\Boroczky}{B{\"o}r{\"o}czky\xspace}

\usepackage{mathcalb}

\newcommand{\areaX}[1]{\cardin{#1}}
\newcommand{\DisksX}[1]{\mathcal{B}_{#1}}

\newcommand{\diskC}{\mathsf{D}}
\newcommand{\diskX}[1]{\diskC\pth{#1}}
\newcommand{\diskY}[2]{\diskC\pth{#1,#2}}

\newcommand{\circY}[2]{\mathcalb{c}\pth{#1,#2}}

\usepackage{booktabs}
\usepackage{siunitx}

\usepackage{booktabs} %
\usepackage[table]{xcolor} %

\definecolor{lightgray}{gray}{0.85}

\usepackage{stmaryrd}%
\providecommand{\IntRange}[1]{\mleft\llbracket #1 \mright\rrbracket}
\newcommand{\IRX}[1]{\IntRange{#1}}%
\newcommand{\IRY}[2]{\left\llbracket #1:#2 \right\rrbracket}
\newcommand{\sIRY}[2]{\llbracket #1:#2 \rrbracket}
\newcommand{\UBlOpt}{\overline{\lOpt}}

\newcommand{\npvC}{\Phi}%
\newcommand{\npvX}[1]{\npvC\pth{#1}}%

\usepackage{nicematrix}
\allowdisplaybreaks

\newcommand{\Poptdef}{\mathsf{P}}

\newcommand{\emaxdef}{\overline{e}}%
\newcommand{\emax}{\LNK{notation:max:e:length}{\emaxdef}}

\newcommand{\Popt}{\LNK{notation:P:opt}{\Poptdef}}

\newcommand{\Term}[1]{\textsf{#1}}

\newcommand{\ArrConf}{\ensuremath{\mathcal{A}_{\mathrm{conf}}}\xspace}
\newcommand{\Hash}{\ensuremath{\mathcal{H}}\xspace}

\newcommand{\DP}{\Term{DP}\xspace}%
\newcommand{\Good}{\mathcal{G}}%
\newcommand{\nG}{{\psi}}%
\newcommand{\dY}[2]{\left\| #1 - #2 \right\|}
\newcommand{\Din}{D_{\mathrm{in}}}
\newcommand{\Dout}{D_{\mathrm{out}}}
\newcommand{\origin}{\mathsf{o}}
\newcommand{\CHX}[1]{\mathcal{CH}\pth{#1}}

\newcommand{\G}{\mathsf{G}}%
\newcommand{\VV}{{V}}%
\newcommand{\EE}{{E}}%
\newcommand{\XSet}{\mathcal{X}}%

\usepackage[most]{tcolorbox}
\usepackage{marginnote}

\newcommand{\NewAuthor}[4]{%
    \expandafter\newcommand\csname #1m\endcsname[1]{%
        \marginnote{\begin{tcolorbox}[
            enhanced, sharp corners, boxrule=0.5pt,
            colback=#2!10, colframe=#2!80!black,
            fontupper=\tiny\sffamily, width=\marginparwidth,
            left=2pt, right=2pt, top=2pt, bottom=2pt
        ] \textbf{#3:} ##1 \end{tcolorbox}}%
    }

    \newenvironment{#1block}[1]
    {\begin{tcolorbox}[
        colback=#2!5, colframe=#2!75!black,
        fonttitle=\bfseries, title=#4's Note: ##1,
        breakable, enhanced
    ]}
     {\end{tcolorbox}}
 \expandafter\newcommand\csname #4\endcsname[2][]{%
    \begin{#1block}{##1}
        ##2
    \end{#1block}
 }

 \expandafter\newcommand\csname #1\endcsname[1]{%
    {\color{#2} \textbf{***} \textbf{#3:} ##1 \textbf{***}}
  }

}

\NewAuthor{sariel}{blue}{SH}{Sariel}

\NewAuthor{qizheng}{red}{QZ}{Qizheng}

\NewAuthor{farouk}{green}{FH}{Farouk}

\usepackage{pgffor}

\newcommand{\gcdot}{\mathbin{\vcenter{\hbox{\scalebox{0.55}{$\bullet$}}}}}
\newcommand{\gcdY}[2]{#1 \gcdot #2}
\newcommand{\pgcdY}[2]{\pth{#1 \gcdot #2}}

\newcommand{\half}{\tfrac{1}{2}}

\usepackage{array} %
\usepackage{etoolbox} %

\newtoggle{show_shortcuts}
\toggletrue{show_shortcuts}

\newcommand{\COIntervalX}[1]{\left[ #1 \right)}

\iftoggle{show_shortcuts}{
  \newcolumntype{H}{l} %
}{
  \newcolumntype{H}{>{\setbox0=\hbox\bgroup}c<{\egroup}@{}} %
}

\makeatletter
\let\c@table\c@figure         %
\let\ftype@table\ftype@figure %
\makeatother                  %

\bibliography{k_grid_points}

\title{How to Catch $k$ Grid Points}

\notSODA{%
\author{%
   Sariel Har-Peled%
   \notSODA{\SarielThanks{}}%
   \and%
   Elfarouk Harb%
   \notSODA{\FaroukThanks{}}%
   \and%
   Qizheng He%
   \notSODA{\QizhengThanks{}}%
}%
}
\date{}

\begin{document}
\maketitle

\begin{abstract}
    Given a positive integer $k$, we study the problem of finding a convex polygon of minimum perimeter that encloses exactly $k$ points of~$\ZZ^2$.  We show that an optimal polygon is contained in a circular annulus of width~$O(k^{1/6})$, has $\Theta(k^{1/3})$ boundary grid points, and its longest edge has length $\Theta(k^{1/4})$.  Using these structural bounds, we present a deterministic algorithm that computes an optimal polygon in $O(k^{29/18+o(1)})$ time, improving over the previous $O(k^3)$-time algorithm.
\end{abstract}

\begin{figure}[H]
    \foreach \i in {1,...,5}{%
       \hfill \includegraphics[page=\i, width=0.18\linewidth]{figs/polygons}%
    }%
    \par\bigskip
    \foreach \i in {6,...,10}{%
       \hfill \includegraphics[page=\i, width=0.18\linewidth]{figs/polygons}%
    }%

    \caption*{The minimum perimeter polygon covering $k$ points, for
       $k=3,\ldots, 12$. More examples are at \figref{rest_1},
       \figref{rest_1b}, \figref{rest_2}, \figref{rest_3}, and
       \figref{rest_4}.}
\end{figure}

\SODA{
\thispagestyle{empty}%
\newpage%
\setcounter{page}{1}%
}
\section{Introduction}
\seclab{intro}

We study the following problem.

\begin{problem}
    Given a positive integer $k$, find a (closed) convex polygon $P$ of minimum perimeter, such that $P$ contains exactly $k$ points of the integer grid $\ZZ^2$.
\end{problem}

By viewing the boundary of $P$ as a rubber band trying to shrink to its minimum length, it is clear that such an optimal polygon has vertices only at points of the grid $\ZZ^2$.  Let $\PoptX{k}$ denote such an optimal polygon, and let $\lOptX{k}$ denote its perimeter.  Observe that $0 = \lOptX{1}<\lOptX{2}<\cdots$ is a strictly increasing monotone sequence, since one can always snip a corner of $\PoptX{k}$ to obtain a shorter-perimeter polygon containing $k-1$ points.

\paragraph{Why is this problem interesting?}

Coming up with a \DP for this problem that is relatively efficient is not obvious. We describe below the known algorithm of Eppstein and Erickson \cite{ee-innfm-94} modified for this case, which has running time $O(k^3)$, which is our baseline. It is natural to ask if one can do better. First, because this is a natural problem. Secondly, as $k$ increases, the optimal solution converges to a disk. There is a lot of related work (see below), some of it recent, trying to understand the behavior of disks on a grid (such as the Gauss circle problem).

Indeed, it is tempting to conjecture that maybe an output sensitive linear-time algorithm exists for this problem (i.e., this would be running time $O(k^{1/3})$ in this case), since the related discrete hull-problem (compute the convex-hull of the grid points inside a given disk of radius $\sqrt{k}$) has such an algorithm \cite{h-osafd-98a}. And yet, to even beat the baseline requires non-trivial ideas and work. One needs to use tools from the geometry of numbers (such as Farey's sequences), and rather intricate exchange arguments to argue about the shape of the optimal solution. Indeed, the \DP uses the structure to improve its performance.

And even if the question is not interesting beyond discrete geometry, we believe the answer is (arguably) quite interesting.  As our work hopefully demonstrates, this problem has a rather intricate and fascinating mathematical structure. Beyond that, scaling our implementation to work on huge inputs required many interesting ideas on the applied side, which should be useful for other \DP problems whose instances are too large to fit in memory.

\paragraph{The starting point: An $O(k^3)$-time algorithm.}

The algorithm of Eppstein and Erickson~\cite{ee-innfm-94}, building on
Dobkin \etal~\cite{ddg-fsp-83}, implies an $O(k^4)$-time algorithm for
this problem.  We sketch how to improve it further to $O(k^3)$, following
the \DP of Eppstein \etal~\cite{eorw-fmakg-92}.  The optimal solution has
perimeter at most $4\sqrt{k}$, so one can restrict the solution to the
grid $G = \sIRY{-2\sqrt{k}}{2\sqrt{k}} \times \sIRY{0}{2\sqrt{k}}$, with
$(0,0)$ being the bottom-right vertex of~$\Popt$.  Anchoring at this
vertex shaves a factor of~$k$ from the running time.

\begin{figure}[t]
    \centering \phantom{}\hfill%
    \includegraphics{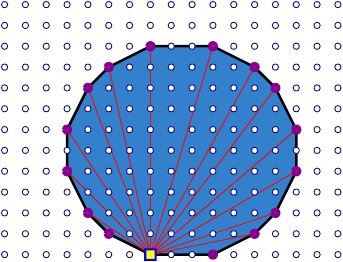} \hfill%
    \includegraphics[page=2]{figs/100} \hfill\phantom{}%
    \caption{Left: An optimal solution for $100$ points and its origin-anchored triangulation. Right: A \DP state consists of an edge from the origin together with the number of grid points covered by the minimum-perimeter partial polygon ending at this edge---in this case $(p, 52)$.}
    \figlab{100}
\end{figure}

The optimal polygon $\Popt$ can be triangulated by connecting all its vertices to the origin, see \figref{100}.  There are $O(k)$ potential such diagonals.  The polygon is reconstructed by gluing these triangles in clockwise order around the origin.  A state of the \DP is a pair $(p, \alpha) \in G \times \IRY{0}{k}$ consisting of the last diagonal from the origin (going counterclockwise), encoded by its endpoint $p$, and the number of grid points $\alpha$ covered so far.  The \DP stores, for each state, the smallest perimeter solution found.  There are $O(k^2)$ states, each with $O(k)$ outgoing transitions (one per triangle).  Each transition updates the perimeter and the grid-point count.  The \DP may consider nonconvex fan chains, but this does not change the optimum: taking the convex hull cannot increase perimeter, and if the hull contains too many grid points we trim it back to exactly $k$ grid points.  As the state space has $O(k^3)$ edges, the optimal solution can be computed in $O(k^3)$ time.

\subsection*{Related work}

\subparagraph*{Isoperimetric inequality.}

Famously, the minimum area shape for a prespecified perimeter is the circle. The problem we are solving is somewhat similar, as the number of grid points inside a (fat) region is a proxy for its area. Some inequalities capture the deficit of a region from being circular.  Bonnesen's inequality \cite{b-sadli-21, f-biidc-91} strengthens the classical isoperimetric inequality---see \thmref{bonnesen} for the exact statement.

\subparagraph*{Polygons and isoperimetry.} %
Macnab showed \cite{m-cprq-81} that if one specifies the length of edges of a polygon, in order, the maximum-area polygon is formed by placing the vertices of the polygon on a circle (such polygons are \emph{cyclic polygons}).  In particular, there is an \emph{isoperimetric inequality for polygons}, see \cite{ftt-apiip-55}---see \thmref{i_p_polygons} for details.

\subparagraph*{Gauss circle problem.} %
Counting the number of grid points in a disk of radius $r$ has a long history. The best known error exponent is $131/208$ due to Huxley \cite{h-eslpi-03}.  Andrews \cite{a-lbvsc-63} proved that the boundary of a strictly convex body of area $\areaC$ contains at most $O(\areaC^{1/3})$ grid points (a body is \emph{strictly convex} if it contains no segment in its boundary). Equivalently, a strictly convex curve of length $\perimC$ passes through at most $O(\perimC^{2/3})$ grid points.

\Barany and Pach \cite{bp-nclp-92} showed that a convex grid polygon with
area $\areaC$ has $O(\areaC^{1/3})$ vertices.  The problem of counting
minimal-perimeter polyominoes of a given area is related---see the work
by Altshuler \etal \cite{ayvwb-mpp-06}.

\subparagraph*{Discrete hull.} %
\Barany and Larman \cite{bl-chip-96} studied the number of faces of various dimensions of the convex-hull of all the grid points contained in a ball of radius $r$ centered at the origin. In two dimensions, they prove that the number of vertices of this discrete hull of the disk of radius $r$ is $\Theta(r^{2/3})$.  \Barany and \Boroczky \cite{bb-lpbih-03} showed that this bound also holds for the number of grid points on the boundary of this polygon.

\subparagraph{Isoperimetry on the grid without convexity.}  Wang and Wang \cite{ww-dip-77} studied the problem of ordering the points of the grid, so that any prefix of the first $k$ points minimizes the number of boundary points (note that they are not considering the convex-hull of the points). The resulting ordering is somewhat similar to Cantor's pairing function that zigzags over $\NN^2$ to map it to $\NN$.

\subparagraph*{The optimal solution is a convex grid polygon.}
The following observation is used implicitly throughout the paper. It formalizes the intuition that a convex set containing at least $k$ grid points can always be trimmed down to a convex polygon containing exactly $k$ of them, without increasing the perimeter.

\begin{observation}
    \obslab{trim-to-k}%
    If a bounded set $Q$ contains at least $k$ grid points, then there is
    a convex grid polygon $Q'\subseteq Q$ that contains exactly $k$ grid points
    and satisfies $\perimX{Q'}\leq\perimX{Q}$.
\end{observation}
\remove{%
\begin{proof}
    Let $S=Q\cap\ZZ^2$.  If $\cardin{S}=k$, there is nothing to prove.
    Otherwise choose a vertex $x$ of $\CHX{S}$ and replace $S$ by
    $S\setminus\{x\}$.  The new hull is contained in the old one, and
    hence its perimeter does not increase.  Moreover, since the new hull
    is contained in $Q$, every grid point in it belongs to $S$; and since
    $x$ was an extreme point, it is not in the new hull.  Thus the number
    of grid points drops by exactly one.  Repeating this step gives the
    desired polygon.
\end{proof}%
}

\subsection*{Our results}

Let $\Popt = \PoptX{k}$ denote an optimal polygon.  Our main contributions are:
\begin{compactenumI}%
    \item \textbf{Bounding Optimal Polygon Characteristics
       (\secref{gap}).} We prove several structural results on the
    optimal polygon $\Popt$. In particular, it is sandwiched between two
    concentric disks of radii $\rin < \Rout$, of width
    \begin{equation*}
        \WD
        =
        \Rout - \rin
        =
        O(k^{1/6}),
    \end{equation*}
    has $b = \Theta(k^{1/3})$ boundary grid points and longest edge $\emax = O(k^{1/4})$ (\corref{edge-improved}).  This also gives the perimeter window
    \[
        2\sqrt{\pi k}-O(k^{-1/6}) \leq \lOptX{k} \leq 2\sqrt{\pi k}.
    \]
    \item \textbf{Dynamic programming algorithms (\secref{dp}).}  We
    describe an edge-vector \DP framework and analyze it under
    successively tighter structural bounds, obtaining algorithms with
    running times $O(k^3)$, $O(k^{2+1/6})$, $O(k^{2})$, $O(k^{2-1/6})$,
    $O(k^{2-2/9})$, and finally $O(k^{29/18+o(1)})$.

    \item{} \textbf{Implementation (\secref{experiments}).} We
    implemented some parts of our algorithm and did extensive
    experimentation to guide our theoretical investigation.  The source
    code is available from github:
    \url{https://github.com/sarielhp/rust_k_perimeter}. A webpage to explore the various optimal solutions is available here:
    \url{https://sarielhp.org/r/26/k_perimeter}.

\end{compactenumI}

\paragraph{Techniques.}
Our proof of the gap bound combines several ingredients:
\begin{compactenumA}
    \item An almost tight upper bound on the perimeter, $\lOptX{k} \leq 2\sqrt{\pi k}$, obtained by an averaging argument over disk centers (\lemref{tight-upper}).

    \item \emph{Bonnesen's inequality} and Pick's theorem, which give $\WD = O(\sqrt{b})$ (\propref{gap-from-b}).

    \item An \emph{edge-structure analysis} of the convex grid polygon, yielding $b = O(k^{3/8}\, \WD^{1/4})$ (\clmref{edge-structure}).
\end{compactenumA}
Bootstrapping~(B) and~(C) gives $\WD = O(k^{3/16}\, \WD^{1/8})$, hence $\WD = O(k^{3/14})$.  A further more careful exchange argument exploiting the optimality of~$\Popt$ (\thmref{improved-gap}) tightens this to~$\WD = O(k^{1/6})$ with~$b = \Theta(k^{1/3})$.  Finally, a number-theoretic analysis of primitive grid vectors (\lemref{p_d_bound}) shows that every primitive edge direction of~$\Popt$ has length~$O(k^{1/4})$, yielding the same bound on the longest edge (\corref{edge-improved}).

\paragraph{Standard tools.}

For readability, the description of some of the standard tools used are
delegated to \secref{background}. These include defining primitive
vectors, \defrefY{generator}{grid polygons}, \thmrefY{pick}{Pick's
   theorem}, \thmrefY{bonnesen}{Bonnesen's inequality},
\defrefY{epicenter}{epicenter}, and \secrefY{farey}{Farey's sequences},
among other tools.

\paragraph*{Notation and setup.}

The notations used here are summarized in \tblref{notations}. Observe that $v$ and $b$ are  different, as $\Popt$ might have grid points in the middle of its edges. For example, a $2\times 2$ square has $v=4, b=8, i=1$ and $k=9 =i+b = 9$. In the following $\ZZ^2$ denotes the integer \emphi{grid}.

\begin{table}[ht]
    \capstart \centering
    \rowcolors{2}{}{lightgray}
    \begin{tabular}{l >{\raggedright\arraybackslash}p{0.68\linewidth}}
      \toprule
      \textbf{Notation}
      & \textbf{Description} \\
      \midrule
      $\ZZ^2$
      & The integer grid
      \\
      $k$
      & Number of points to be enclosed
      \\%-------------------------------------------------------------------
      $\Popt = \PoptX{k}$\notlab{P:opt}
      & Optimal minimum perimeter polygon enclosing $k$ grid points
      \\%--------------------------------------------------------------
      $\Lopt = \lOptX{k}$\notlab{L:opt} %
      & Optimal perimeter: $\Lopt = \perimX{\Popt}$ %
      \\%-------------------------------------------------------
      $\areaOpt = \areaX{\Popt}$
      & Area of $\Popt$
      \\%-------------------------------------------------------
      $\bbdef$\notlab{b}
      & Number of boundary grid points on $\partial \Popt$
      \\%-------------------------------------------------------
      $\nvdef$\notlab{v}
      & Number of vertices of $\partial \Popt$
      \\%-------------------------------------------------------
      $\iidef$\notlab{i}
      &
        Number of grid points in the interior of $\Popt$
      \\%-------------------------------------------------------
      $\WD$\notlab{g_mark}
      & Width of the min-width disk sandwich containing $\Popt$ (\defref{epicenter})
      \\%-------------------------------------------------------
      $\epi$
      &
        Center of the min-width disk sandwich containing $\Popt$ (\defref{epicenter})%
      \\%-------------------------------------------------------
      $\rin$/$\Rout$\notlab{r:in}\notlab{r:out}
      & Radii with $\diskY{\epi}{\rin}\subseteq\Popt\subseteq\diskY{\epi}{\Rout}$
      \\%-------------------------------------------------------
      $\emaxdef$\notlab{max:e:length}
      & Upper bound on  max length edge of $\Popt$
      \\%-------------------------------------------------------
      $\gcdY{x}{y}$
      &
        $\gcd( |x|, |y| )$
      \\%-------------------------------------------------------
      \bottomrule
    \end{tabular}
    \caption{Summary of notations. Observe that $k = i + b$.}
    \tablab{notations}
\end{table}

\paragraph{Paper organization.}
We begin in \secref{bounds} by establishing rough bounds on the optimal
perimeter $\lOpt$ and on the radii of the minimum-width disk sandwich
containing $\Popt$.
\secref{gap} forms the technical heart of the structural analysis: by
combining Bonnesen's inequality, an edge-structure argument, and a careful
exchange argument, we show that $\Popt$ is sandwiched in an annulus of
width $\WD = O(k^{1/6})$, has $b = \Theta(k^{1/3})$ boundary grid points,
longest edge $O(k^{1/4})$, and maximum turn angle $O(k^{-1/6})$.
\secref{dp} then develops our dynamic-programming algorithm through a
sequence of successively faster variants---plugging the structural bounds
into the natural \DP, restricting the search to a thin annulus, cutting the
annulus into thinner rings, exploiting the small turn angle, reducing the
epicenter location uncertainty, and finally batching
transitions via monotone queues (\secref{monotonicity})---culminating in
the $O(k^{29/18+o(1)})$-time algorithm.  \secref{experiments} describes our
Rust and C++ implementation, the heuristics that make it practical, and
experiments solving instances as large as $k = 2.5 \cdot 10^{6}$.  For
readability, the standard tools used throughout---primitive vectors, Pick's
theorem, Bonnesen's inequality, the epicenter, and Farey
sequences---are collected in \secref{background}.  Additional examples of
optimal $k$-polygons are given in the appendix.

\section{Background}
\seclab{background}

\subsection{Classical ingredients}

A {vector} $(x,y) \in \ZZ^2$ is \emphi{primitive} if $\gcdY{x}{y} = \gcd(|x|,|y|) = 1$.

\begin{defn}
    \deflab{generator}%
    A \emphi{grid polygon} is a connected polygon with vertices on the grid.  The boundary of a grid polygon $P$ is considered to be oriented in a counterclockwise direction. An oriented edge $e$ from $u$ to $v$ has the associated direction vector $v - u$. The unique primitive vector $\tau$ such that $ v - u = \alpha \tau$, for some positive integer $\alpha > 0$, is the \emphi{generator} of $e$.
\end{defn}

\begin{theorem}[Pick's]
    \thmlab{pick}%
    For a simple lattice polygon $P$, with $i \geq 0$ interior grid
    points, $b$ boundary grid points, and area $\areaC$, we have
    $\areaC = i + b/2 - 1$.
\end{theorem}

A consequence of Pick's theorem is that all grid triangles with no grid points in their interior, or the interior of their edges, have area $1/2$. These are \emphw{primitive} triangles. Primitive vectors form the edges of a primitive triangle. Any two vectors appearing on the boundary of a primitive triangle span the integer grid. See \cite{hw-itn-08} for more details.

\begin{theorem}[Bonnesen's inequality]
    \thmlab{bonnesen}%
    Every bounded convex polygon $P$ satisfies
    \begin{equation*}
        \perimC^2 - 4\pi \areaC
        \geq
        \pi^2\pth{R-r}^2.
    \end{equation*}
    where $\areaC$ denotes the area, $\perimC$ the perimeter, $r$ the inradius, and $R$ the circumradius of $P$.
\end{theorem}

\begin{remark}
    We use \thmref{bonnesen} only in its standard inradius/circumradius
    form.  The conversion from those possibly different centers to the
    common-center disk sandwich used by the algorithms is proved in
    \propref{lower-bonnesen}.
\end{remark}

\begin{defn}
    \deflab{epicenter}%
    Let $(\epidef,\rin,\Rout)$ minimize $\Rout-\rin$ subject to
    \[
        \diskY{\epi}{\rin}\subseteq\Popt\subseteq\diskY{\epi}{\Rout}.
    \]
    The point $\epi$ is the \emphi{epicenter}, and
    $\WD=\Rout-\rin$ is the annulus width of $\Popt$.
\end{defn}

\remove{%

}

\begin{theoremnp}[\Barany--Pach] \hcite{bp-nclp-92} %
    \thmlab{Andrews}%
    A convex grid polygon of perimeter $\perimC$ has at most
    $O(\perimC^{2/3})$ vertices.
\end{theoremnp}

\begin{fact}
    \factlab{approx}%
    For $x > 0$, we have $x + \frac{1}{2x} - \frac{1}{8x^3} \leq \sqrt{1+x^2} \leq x + \frac{1}{2x}$. For $x \in (0,1/2)$, we have $1-x^2 \leq 1-\tfrac{x^2}{2} - \tfrac{x^4}{8} \leq \sqrt{1-x^2} \leq 1-\tfrac{x^2}{2}$.
\end{fact}

\subsection{A short foray into Farey's sequences}
\seclab{farey}

The following is standard introductory material in number theory---see Hardy and Wright \cite{hw-tn-65}.  A rather elegant construction of all primitive vectors up to a certain length follows from the following surprisingly elegant algorithm. Let $S_1 = (1,0), (0,1), (-1,0), (0,-1)$ be the (cyclical) sequence of the four primitive vectors with a zero coordinate, sorted in counterclockwise order. The $i$\th sequence $S_i$ is generated from $S_{i-1}$ by inserting between any two consecutive vectors in the original sequence their sum. Thus,
\begin{equation*}
    S_2 = (1,0), \underline{(1,1)}, (0,1), \underline{(-1,1)}, (-1,0), \underline{(-1,-1)}, (0,-1), \underline{(1,-1)}.
\end{equation*}
These sequences $S_1,S_2, \ldots$ are equivalent to \emphi{Farey sequence}. Every primitive vector $u$ has an index $i$ such that $S_i$ contains $u$.

\begin{fact}
    \factlab{primitive_number}%
    For an integer $R >0$, let $\Lambda_R$ be the set of all primitive vectors in $[-R, R]^2$. Then:
    \begin{compactenumI}
        \item $\Lambda_R = S_{R+1} \cap [-R,R]^2$.
        \item The number of primitive grid vectors in $[-R,R]^2$ is $\cardin{\Lambda_R} = \Theta(R^2)$.
    \end{compactenumI}
\end{fact}

The two \emphi{parents} of a primitive vector $u \in S_i \setminus S_{i-1}$, for $i >1$, are the two (consecutive) unique primitive vectors $u_1, u_2 \in S_{i-1}$, such that $u = u_1 + u_2$.

\begin{fact}
    For a primitive vector $u$ with parents $u_1, u_2$, we have that $\lenX{u} > \lenX{u_1}$ and $\lenX{u} > \lenX{u_2}$.
\end{fact}

The $i$\th frame is the set $F_i = \Lambda_i \setminus
\Lambda_{i-1}$. Clearly, $|F_i| \leq 8i$.
\begin{claim}
    \clmlab{shorty}%
    A vector $v \in F_i$ can be the parent of at most $2R/i$ vectors in $\Lambda_R$.
\end{claim}
\begin{proof}
    Such a vector is the parent of at most two new children in $S_j$, for all $j > i$. This readily gives the upper bound $O(R)$. To get the better bound, we need to work slightly harder. For simplicity of exposition, assume $i>1$, $v =(i,\alpha)$, with $0 < \alpha < i$. Consider the children of $v$ that (say) are clockwise to it: $c_1, c_2, \ldots ,c_t$ in their order of generation, with $c_j = (x_j, y_j)$. Observe that $x_1 > i$ (as the vector $(1,0)$ is clockwise to it), and in particular, by the above construction we have the invariant that $x_{j+1} = x_j + i$. But, after $t = R/i$ times this happens, $X_{t+1} > R$, which is impossible.
\end{proof}

\begin{lemma}
    \lemlab{long}%
    For any integer $R > 1$ and any fixed $\eta>0$, there exists a
    constant $c = c(\eta)>0$ such that all but at most $\eta R^2$
    vectors in $\Lambda_R$ have the property that both their parents are
    of length at least $c R$. Such primitive vectors are
    \emphi{$cR$-long}.
\end{lemma}
\begin{proof}
    By \clmref{shorty}, for $\beta >0$, the vectors in $\Lambda_{\beta}$
    can be the parents of at most
    \begin{equation*}
        \sum_{j = 1}^{\beta}  8j \cdot \frac{2R}{j} = 16 R \beta
    \end{equation*}
    vectors in $\Lambda_R$.  Set $\beta=cR$, where $c>0$ is small
    enough that $16cR^2\leq \eta R^2$.  Then at most $\eta R^2$ vectors
    of $\Lambda_R$ have a parent in $\Lambda_{cR}$, proving the claim
    after adjusting the constant to pass from $\ell_\infty$ length to
    Euclidean length.
\end{proof}

\begin{observation}
    \obslab{long_are_good}%
    Consider a $cR$-\lemrefY{long}{long} primitive vector
    $u \in \Lambda_R$, with parents $u_1$ and $u_2$.  The projections of
    $u_1$ and $u_2$ on the segment $\origin u$ have length
    $\Omega_c(R)$, and their distances from the other endpoint of the
    segment are also $\Omega_c(R)$.
\end{observation}
\begin{proof}
    The vector $u$ lies between its two parents in angular order, so the
    orthogonal projections of $u_1$ and $u_2$ onto the line spanned by
    $u$ point in the direction of~$u$, and their projected lengths add
    to $\lenX{u}$.  Each parent has Euclidean length at least~$cR$ and
    is at perpendicular distance $1/\lenX{u}\leq 1$ from this line,
    because the corresponding primitive triangle has area~$1/2$.  Thus
    each projected length is $\Omega_c(R)$.  Since these two projected
    lengths add to $\lenX{u}$, the distance from each projected point to
    the other endpoint of $\origin u$ is also $\Omega_c(R)$.
\end{proof}

\section{Bounds on the optimal perimeter}
\seclab{bounds}

We start by proving some rough bounds on the optimal perimeter $L_k$ and
the radius of the minimum-width disk sandwich containing $\Popt$. This is
used in \secref{gap}.

\subsection{Upper-bound via density}

For $r \geq 1/\sqrt{2}$, let $\DisksX{r} = \Set{ \diskY{p}{r}}{ p \in \ZZ^2}$ be the cover of the plane by disks of radius $r$, centered at the points of the grid. Let $\diskX{R} = \diskY{0}{R}$ denote the disk of radius $R$ centered at the origin. For a set $X \subseteq \Re^2$, its \emphi{area} is denoted by $\areaX{X}$.

\begin{lemma}
    \lemlab{tight-upper}%
    $\lOpt = \LoptX{k} \leq 2\sqrt{\pi k}$.
\end{lemma}

\begin{proof}
    Let $r=\sqrt{k/\pi}$, and for $p\in[0,1)^2$ let
    \[
        N(p)=\cardin{\ZZ^2\cap\diskY{p}{r}}.
    \]
    By Fubini, averaging over one fundamental cell gives
    \[
        \int_{[0,1)^2} N(p)\,dp
        =
        \sum_{z\in\ZZ^2}
        \areaX{\bigl([0,1)^2\cap\diskY{z}{r}\bigr)}
        =
        \areaX{\diskX{r}}
        =
        \pi r^2
        =
        k.
    \]
    Hence some translate center $p$ satisfies $N(p)\geq k$.

    Let $S = \ZZ^2 \cap \diskY{p}{r}$. Clearly, $|S| \geq k$. Let $P = \CHX{S}$ be the convex hull of $S$. This is a convex grid polygon with $\perim(P) \leq \perim(\diskY{p}{r}) = 2\pi r$, since $P \subseteq \diskY{p}{r}$ and for nested convex bodies the inner one has smaller perimeter. By \obsref{trim-to-k}, there is a smaller polygon $P$, with smaller perimeter, containing exactly $k$ points and with perimeter $\lOpt \leq \perimX{P} \leq 2\pi r = 2 \sqrt{\pi k}$.
\end{proof}

\begin{proposition}
    \proplab{edge-bound}%
    If a convex polygon $P$ satisfies
    $\diskY{p}{\rin}\subseteq P\subseteq\diskY{p}{\Rout}$, with
    $\WD=\Rout-\rin$, then the longest edge $e$ of $P$ satisfies
    $\lenX{e} \leq 2\sqrt{2 \WD \Rout}$.  In particular, for
    $P=\Popt$ and its minimum-width disk sandwich,
    $\lenX{e}=O(\sqrt{\WD}\,k^{1/4})$.
\end{proposition}
\begin{proof}
    \begin{figure}[H]
        \centering \includegraphics{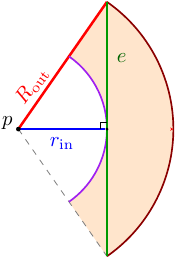}
        \caption{}
        \figlab{annulus}
    \end{figure}
    The supporting line of every edge avoids the interior of
    $\diskY{p}{\rin}$, while its endpoints lie in $\diskY{p}{\Rout}$.
    The maximum length such edge $e$ has its endpoints on the outer
    circle $\circY{p}{\Rout}$ and is tangent to the inner circle
    $\circY{p}{\rin}$, see \figref{annulus}. Thus, we have
    $\Rout^2 = \rin^2 + \lenX{e}^2/ 4$, which implies
    \begin{math}
        \lenX{e}%
        =%
        2 \sqrt{\Routdef^2 - \rindef^2 }%
        =%
        2 \sqrt{\Routdef^2 - {(\Rout - \WD)}^2 } \leq 2 \sqrt{2 \Rout \WD }.\Bigr.
    \end{math}
    For $P=\Popt$, \clmref{sandwich-radius-rough} gives
    $\Rout=O(\sqrt{k})$, yielding the stated asymptotic bound.
\end{proof}

\subsection{Lower bounds}

\begin{observation}
    \proplab{lower-trivial}%
    Consecutive boundary grid points along $\partial \Popt$ being at distance $\geq 1$ from each other implies that $\NTX{\lOpt} \geq b$.
\end{observation}

\begin{claim}
    \clmlab{lower_1}%
    For $k$ sufficiently large, if $\bb \leq c k^\alpha$, for some constants $\alpha \in (0,1/2]$ and $ c \geq 1$, then $\lOpt = \perimX{\Popt} \geq 2 \sqrt{\pi k} - 2 \sqrt{\pi} c k^{\alpha - 1/2}$.  Specifically, we have
    \begin{math}
        \NTX{\lOpt} \geq 2 \sqrt{\pi k} -4.
    \end{math}
\end{claim}

\begin{proof}
    By \thmrefY{pick}{Pick's theorem}, we have that
    \begin{math}
        \NTX{\areaOpt} = \ii + \bb/2 - 1 =%
        \ii + \bb - \bb/2 - 1 \geq k - ck^\alpha.
    \end{math}
    By \factref{approx} and the isoperimetric inequality, we have that
    \begin{align*}
        \perimX{\Popt}
        &\geq
        2\pi \sqrt{\areaOpt/\pi}
        \geq
        2\sqrt{\pi k}\sqrt{1-c/k^{1-\alpha }}
        \\&
        \geq
        2\sqrt{\pi k}\pth{1-c/k^{1-\alpha }}
        =%
        2 \sqrt{\pi k} - 2 \sqrt{\pi} c k^{\alpha - 1/2}.
    \end{align*}

    By \lemref{tight-upper}, $\lOpt \leq 2\sqrt{\pi k}$ and \propref{lower-trivial}, we have that $b \leq 2\sqrt{\pi k}$; plugging this into the above argument and being slightly more careful with the constant implies the second bound.
\end{proof}

\corref{perimeter-window} sharpens the above lower bound to $2\sqrt{\pi k} -O(1/k^{1/6})$.

\begin{proposition}
    \proplab{lower-bonnesen}%
    We have
    \begin{math}
        \lOpt \geq \sqrt{4\pi(k - b/2 - 1) + (\pi^2/4) \WD^2}.
    \end{math}
\end{proposition}

\begin{proof}
    Let $r_*$ be the inradius of $\Popt$ and let $R_*$ be its
    circumradius.  Choose an incenter $a$ and a circumcenter $b_*$, so
    that
    \[
        \diskY{a}{r_*}\subseteq\Popt\subseteq\diskY{b_*}{R_*}.
    \]
    Since $\diskY{a}{r_*}\subseteq\diskY{b_*}{R_*}$, we have
    $\dY{a}{b_*}\leq R_*-r_*$.  Therefore
    \[
        \Popt\subseteq\diskY{b_*}{R_*}
        \subseteq
        \diskY{a}{R_*+\dY{a}{b_*}}
        \subseteq
        \diskY{a}{r_*+2(R_*-r_*)}.
    \]
    Thus the minimum-width concentric disk sandwich of $\Popt$ has width
    $\WD\leq 2(R_*-r_*)$.  Applying
    \thmrefY{bonnesen}{Bonnesen's inequality} to $r_*,R_*$ gives
    \begin{math}
        \lOpt^2 - 4\pi \areaOpt \geq \pi^2 (R_*-r_*)^2
        \geq
        (\pi^2/4)\WD^2.
    \end{math}
    Rearranging, we have
    \begin{math}
        \lOpt \geq \sqrt{4\pi \areaOpt+(\pi^2/4)\WD^2},
    \end{math}
    and plugging the bound of \thmrefY{pick}{Pick's theorem},
    \begin{math}
        \areaOpt = i + b/2 - 1 = k - b / 2 - 1,
    \end{math}
    implies the bound.
\end{proof}

\begin{claim}
    \clmlab{sandwich-radius-rough}
    The radii of the minimum-width disk sandwich satisfy
    $\rin=O(\sqrt{k})$ and $\Rout=O(\sqrt{k})$.
\end{claim}
\begin{proof}
    Since $\diskY{\epi}{\rin}\subseteq\Popt$, we have
    $\pi\rin^2\leq\areaOpt\leq k$ by Pick's theorem, and hence
    $\rin=O(\sqrt{k})$.
    Also $\epi\in\Popt$.  For any convex body, its diameter is at most
    half its perimeter, because two boundary points at diameter distance
    split the boundary into two curves, each of length at least the
    diameter.  Thus
    \[
        \operatorname{diam}(\Popt)\leq \lOpt/2=O(\sqrt{k})
    \]
    by \lemref{tight-upper}.  Since every point of $\Popt$ is within
    $\operatorname{diam}(\Popt)$ of $\epi$, we have $\Rout=O(\sqrt{k})$.
\end{proof}

\section{Bounding optimal polygon characteristics}
\seclab{gap}

Our algorithm and analysis rely on bounding certain quantities that characterize the optimal polygon. We use the following edge-vector terminology throughout this section:
a vector $u=(x,y)\in\ZZ^2$ is \emphi{primitive} if
$\gcdY{x}{y}=1$;\footnote{See \tabref{notations} for notations used, in particular $\gcdY{x}{y}=\gcd(|x|, |y|)$} if an oriented edge $e$ of $\Popt$ has direction
vector $e=m u$, where $u$ is primitive and $m\in\ZZ_{>0}$, then we say
that $u$ \emphi{generates} $e$ and call $m$ the \emphi{multiplicity} of
$u$ on $e$.

The main quantities are:
\begin{compactenumI}
    \item The width of the minimum-width concentric disk sandwich
    $\diskY{\epi}{\rin}\subseteq\Popt\subseteq\diskY{\epi}{\Rout}$.
    This restricts the locations the \DP has to consider for vertices of
    $\Popt$.

    \item The longest edge and longest primitive edge of $\Popt$. This determines how far the \DP has to ``look'' in updating states.
\end{compactenumI}
The bounds established in this section are summarized in \tblref{summary}.

\begin{table}[t]
    \centering
    \definecolor{tableGray}{RGB}{232, 234, 237} \centering
    {%
       \begin{NiceTabular}{ll l}
           \CodeBefore \rowcolors{2}{tableGray}{white} \Body \toprule \textbf{Quantity} & \textbf{Best Bound} & \textbf{Source}
           \\
           \midrule Boundary grid points & $b = \Theta(k^{1/3})$ & \thmref{improved-gap}
           \\
           Annulus width & $\WD = O(k^{1/6})$ & \thmref{improved-gap}
           \\
           Number of vertices & $\nv = \Theta(k^{1/3})$ & \Block[l]{}{
              \thmref{Andrews}, \\
              \corref{nv_bound} }%
           \\
           Perimeter & $2\sqrt{\pi k} - O(k^{-1/6}) \leq \lOpt \leq 2\sqrt{\pi k}$ & \corref{perimeter-window}
           \\
           Primitive edge length & $\lenX{u} = O(k^{1/4})$ & \lemref{p_d_bound}
           \\
           Longest edge & $\emax = O(k^{1/4})$ & \corref{edge-improved}
           \\
        Max. turn angle of $\Popt$ & $\alpha=O(1/k^{1/6})$ & \lemref{turn_angle_2}
           \\
           \bottomrule
       \end{NiceTabular}
    }
    \caption{Summary of bounds established.}
    \tbllab{summary}
\end{table}
The maximum turn angle of $\Popt$ is $O(1/k^{1/6})$.

\subsection{Initial bounds}

\begin{observation}
    \obslab{b_ub_silly}%
    By \lemref{tight-upper}, $b \leq \lOpt \leq 2 \sqrt{\pi k}$, as the distance between two boundary vertices is at least $1$.
\end{observation}

\begin{proposition}
    \proplab{gap-from-b}
    \begin{math}
        \NTX{\WD} \leq \sqrt{\frac{8\bb + 16}{\pi}}
    \end{math}
    and thus $\WD = O(\sqrt{\bb})$ and $\WD = O(k^{1/4})$.
\end{proposition}

\begin{proof}
    By \lemref{tight-upper} and \propref{lower-bonnesen}, we have
    \begin{equation*}
        4\pi k
        \geq
        \lOpt^2
        \geq
        4\pi k - 2\pi b - 4\pi + (\pi^2/4) \WD^2
        \implies
        \frac{8\pi \bb + 16\pi}{\pi^2}
        \geq
        \WD^2.
    \end{equation*}
    Thus $\WD^2=O(b+1)$.  By \obsref{b_ub_silly},
    $b=O(\sqrt{k})$, and hence $\WD=O(k^{1/4})$.
\end{proof}

\begin{claim}
    $\Rout \geq \sqrt{k/\pi} - O(1)$ and $\rin \geq \sqrt{k/\pi} - O(k^{1/4}) = \Omega(\sqrt{k})$.
\end{claim}
\begin{proof}
    We have $2\pi \Rout \geq \lOpt $, and by \clmref{lower_1}, we have that
    \begin{math}
        2\pi \Rout \geq 2 \sqrt{\pi k} -O(1),
    \end{math}
    and thus
    \begin{math}
        \Rout \geq \sqrt{ k/ \pi} -O(1).
    \end{math}
    The bound now follows by \propref{gap-from-b}.
\end{proof}

\begin{observation}
    Clearly, $\areaOpt = \Theta(k)$. Intuitively, we expect the boundary of $\Popt$ to have few grid points in the middle of its edges, as $\Popt$ is relatively round. Thus, the \Barany--Pach theorem for convex grid polygons \cite{bp-nclp-92} suggests that the number of vertices of $\Popt$ should be $\Theta(\areaOpt^{1/3}) = \Theta( k^{1/3} )$. Thus, by \propref{gap-from-b}, $\WD$ should be $O(k^{1/6})$. We prove these are indeed the correct guesses.
\end{observation}

\begin{claim}
    \clmlab{edge-structure}%
    Using the notations of \tabref{notations}, we have $\bb = O(k^{3/8}\, \WD^{1/4})$.
\end{claim}

\begin{proof}
    We may assume that $\Popt$ has no consecutive collinear edges, since deleting any vertex lying between two collinear edges does not change the polygon, its perimeter, or its boundary count $b$.

    Let $p_1, \ldots, p_s$ be the vertices of $\Popt$, as one traverses the boundary of $\Popt$ in counterclockwise orientation. Consider the vectors forming the edges $E = \{ e_1,\ldots,e_s \}$ of $\Popt$, where $e_i = p_i - p_{i-1}$ (with $p_0 = p_s$). Since $\Popt$ is a grid polygon, each edge has the form
    \begin{math}
        e_t = m_tu_t,
    \end{math}
    where $u_t =(x_t,y_t) \in \ZZ^2$ is primitive (i.e., $\gcdY{ x_t}{y_t} = 1$), and $m_t> 0$ is a positive integer. Thus, $\lenX{e_t}=m_td_t$, where $d_t=\lenX{u_t}$, and
    \begin{math}
        \lOpt = \sum_{t=1}^s m_td_t.
    \end{math}
    Also, because a segment with primitive direction $u_t$ and multiplier $m_t$ contains exactly $m_t+1$ grid points, summing over edges gives
    \begin{math}
        b=\sum_{t=1}^s m_t.
    \end{math}
    Next, let $u_1,u_2,\dots,$ be an enumeration of all primitive grid vectors in nondecreasing order of norm (using lexicographical ordering to order equal length vector).  The number of primitive grid vectors of length at most $R$ is $\Theta(R^2)$ \cite{hw-tn-65}. Thus, we have
    \begin{equation*}
        \rho_j
        =
        \lenX{ u_j }
        =
        \Theta(\sqrt j).
    \end{equation*}
    For any $j$, let $\Lambda_j$ be the length of the vector in $E$, having the direction $u_j$. If there is no such vector, then $\Lambda_j = 0$.  Hence
    \begin{equation*}
        \bb
        =
        \sum_{j}\frac{\Lambda_j}{\rho_j}
        \qquad\text{and}\qquad%
        \lOpt=\sum_{j\ge1}\Lambda_j.
    \end{equation*}
    By \propref{edge-bound}, every edge length is at most
    \begin{math}
        \emax = O(\sqrt{\WD} k^{1/4}),
    \end{math}
    and thus for all $j$, we have
    \begin{math}
        0\le \Lambda_j\le \emax.
    \end{math}
    Since $1/\rho_j$ is nonincreasing in $j$, to maximize $b$, we should use the shortest primitive directions first, using maximum edge length $\emax$.  To this end, let $J = \ceil{ \lOpt / \emax}$, set $\Lambda_i = \emax$, for $i=1,\ldots J-1$, and $\Lambda_J = \lOpt - \sum_{i=1}^{J-1} \Lambda_i $.  As $\lOpt=O(\sqrt{k})$, this implies that
    \begin{align*}
        \bb
        &\leq
        \emax \sum_{j=1}^J \frac1{\rho_j}
        =
        O\pth{ \Bigl. \emax \smash{\sum_{j=1}^J \frac{1}{j^{1/2}} }    }
        =
        O( \emax\sqrt{J} \,) \\
        &=
        O( \emax\sqrt{\lOpt / \emax} \,)
        =
        O( \sqrt{ \emax \lOpt } \, )
        =
        O(  \WD^{1/4} k^{1/8} \cdot {k}^{1/4} ).
    \end{align*}
\end{proof}

\begin{lemma}
    \lemlab{main-gap}%
    Using the notations of \tabref{notations}, we have $\WD = O(k^{3/14})$.
\end{lemma}

\begin{proof}
    Combining \propref{gap-from-b} and \clmref{edge-structure}, we have
    \begin{equation*}
        \WD
        =
        O\pth{ \smash{ \sqrt{b}} \bigr.\ts }
        =
        O\pth{ \! \sqrt{k^{3/8}\, \WD^{1/4}}\ts },
    \end{equation*}
    and thus $\WD = O(k^{3/14})$.
\end{proof}

\begin{observation}
    We have
    \[
        b = O(k^{3/8}\WD^{1/4}) = O( k^{3/8+3/56} ) = O(k^{3/7}).
    \]
\end{observation}

\begin{corollary}
    \corlab{edge}%
    The longest edge of $\Popt$ has length at most $\emax= O( k^{5/14})$.
\end{corollary}

\begin{proof}
    Consider the longest edge $e$ of $\Popt$.  By \propref{edge-bound}, we have
    \begin{equation*}
        \lenX{e}
        =%
        O(\sqrt{ \WD} {k}^{1/4} )
        =%
        O( k^{5/14}).
    \end{equation*}
\end{proof}

\begin{corollary}
    $\Popt$ has $\Omega(k^{1/7})$ vertices.
\end{corollary}

\begin{proof}
    By \corref{edge}, we have that number of vertices is at least $\lOpt / \emax = \Omega(\sqrt{k} / k^{5/14}) = \Omega(k^{1/7})$.
\end{proof}

\subsection{Improved bounds via an exchange argument}

The previous bounds are a good start, but can be significantly improved using a careful exchange argument.

The following is the \emph{isoperimetric inequality for polygons}.
\begin{theoremnp}\hcite{ftt-apiip-55} %
    \thmlab{i_p_polygons}%
    Let $P$ be a polygon with $\nv$ sides with perimeter $\perimC$ and area $\areaC$.  We have that
    \begin{math}
        \perimC^2 - 4 \areaC \nv \tan \frac{\pi}{\nv} \geq 0,
    \end{math}
    with equality only if $P$ is a regular $\nv$-gon.
\end{theoremnp}

We also list a few helpful technical lemmas.

\begin{lemma}
    \lemlab{vertex-lb}%
    The optimal polygon $\Popt$ has $\nv \geq \sqrt{k/b}$ vertices.
\end{lemma}

\begin{proof}
    Consider the optimal polygon $\Popt$, which has minimum perimeter $\lOpt$ and area $\areaOpt$.  By \thmref{i_p_polygons}, we have
    \begin{equation*}
        \areaOpt
        \leq
        \frac{\lOpt^2}{4\nv \tan \tfrac{\pi}{\nv}}
        \implies%
        \lOpt^2 \geq
        \areaOpt \cdot 4\nv \tan \tfrac{\pi}{\nv}.
    \end{equation*}
    By the Taylor's expansion of $\tan$, we have $\tan(x) \geq x + \tfrac{x^3}{3}$, for $x \in \COIntervalX{0, \pi/2}$. This implies
    \begin{equation*}
        4\nv \tan \frac{\pi}{\nv}
        \geq
        4\nv\Bigl(\frac{\pi}{\nv} + \frac{\pi^3}{3\nv^3}\Bigr)
        =
        4\pi + \frac{4\pi^3}{3\nv^2}.
    \end{equation*}
    By \thmrefY{pick}{Pick's theorem}, $\areaOpt = k - \tfrac{b}{2} - 1 \leq k$.  Since $b \leq 2\sqrt{\pi k}$ (\obsref{b_ub_silly}), we have $\areaOpt \geq k/2$ for all sufficiently large $k$.  By \lemref{tight-upper}, we have
    \begin{align*}
        &
        4\pi k
        \geq
        \lOpt^2
        \geq
        \areaOpt \Bigl(4\pi + \frac{4\pi^3}{3\nv^2}\Bigr)
        \implies%
        k
        \geq
        \areaOpt \Bigl(1 + \frac{2}{\nv^2}\Bigr)
        \implies
        \frac{k - \areaOpt }{ \areaOpt }
        \geq
        \frac{2}{\nv^2}.
    \end{align*}
    As $\bb \geq \bb/2+1 = k - \areaOpt$, and $\areaOpt \geq k/2$, we have that
    \begin{math}
        \frac{\bb}{k/2} \geq \tfrac{k - \areaOpt }{ \areaOpt } \geq \tfrac{2}{\nv^2},
    \end{math}
    implying that $\nv^2 \geq {k}/{\bb}$.
\end{proof}

\begin{figure}[ht]
    \centering \includegraphics{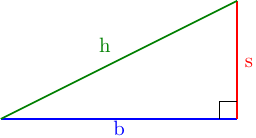}
    \caption{}
    \figlab{triangle}
\end{figure}

\begin{claim}
    \clmlab{triangle}%
    Consider a right triangle with legs $b>0$ and $s$, and hypotenuse $h$. Then
    \begin{math}
        s^2/(2h) \leq h-b \leq s^2/(2b),
    \end{math}
    see \figref{triangle}.
\end{claim}

\begin{proof}
    Since $h=\sqrt{b^2+s^2}$, we have
    \begin{equation*}
        h-b
        =
        \sqrt{b^2+s^2}-b
        =
        \frac{s^2}{\sqrt{b^2+s^2}+b}.
    \end{equation*}
    As $b\leq h$, the denominator is at most $2h$. As $b>0$, it is at least $2b$. Thus
    \begin{equation*}
        \frac{s^2}{2h}\leq h-b\leq \frac{s^2}{2b}.
    \end{equation*}
\end{proof}

\secref{farey} provides a short foray on Farey sequences, used in proving
the following.

\begin{lemma}
    \lemlab{long_expand}%
    Let $R >0$ be an integer, let $c \in (0,1)$ be a constant, and let
    $P$ be a convex grid polygon.  Let $\mathcal{E}$ be a family of
    pairwise nonadjacent edges of $P$ such that the generator $u$ (i.e. the primitive vector) of
    each edge satisfies $\lenX{u}\leq R$ and is
    $cR$-\lemrefY{long}{long}.  For every integer
    $z\leq\cardin{\mathcal{E}}$, one can compute a new grid polygon $Q$
    that contains at least $z$ more grid points, and
    $\perimX{Q} \leq \perimX{P} + O_c(z/R^3)$.
\end{lemma}
\begin{proof}
    Choose any $z$ edges from $\mathcal{E}$.  Consider one such oriented
    edge, with generator~$u$, and let $u_1,u_2$ be the two Farey parents
    of~$u$.  Since $u=u_1+u_2$ and consecutive Farey vectors form a
    unimodular basis, the parents lie on the two lattice lines adjacent
    to the line spanned by~$u$, at perpendicular distance
    $1/\lenX{u}$ from it.  One of these two parents is on the exterior
    side of the boundary edge; call it $u_i$.

    Pick a primitive subsegment $[p,p+u]$ of the edge and add the grid
    point $w=p+u_i$.  By \obsref{long_are_good}, the orthogonal
    projection of~$w$ to the edge is at distance $\Omega_c(R)$ from both
    endpoints of this primitive subsegment.  In particular,
    $\lenX{u}=\Omega_c(R)$, and the height of~$w$ above the edge is
    $1/\lenX{u}=O_c(1/R)$.  Replacing the original full edge by the two
    segments through~$w$ increases the perimeter by
    $O_c(R^{-3})$, by applying \clmref{triangle} to the two thin right
    triangles.

    Do this independently for the chosen pairwise nonadjacent edges.  The
    added exterior grid points are distinct: for each such point, the
    orthogonal projection to $P$ lies in the relative interior of the edge
    on which it was inserted, and this projection is the unique closest
    point of $P$ because the edge is a supporting face.  The same exterior
    point therefore cannot be charged to two different edges.  Now take
    the convex hull of the resulting polygonal curve.  Taking the convex
    hull can only decrease perimeter, so the total increase is
    $O_c(z/R^3)$.
\end{proof}

Informally, the following technical lemma bounds how much extra perimeter we have to ``pay'' to cover $z$ more grid points, assuming the optimal polygon $\Popt$ has an edge with a primitive generator vector of length $d$.

\begin{lemma}
    \lemlab{add}%
    Consider an edge $e=p'q'$ of a convex grid polygon $P$, such that $e$ is formed by $m> 8$ copies of a primitive vector $u$ of length $d$.  For every integer $z$ with $2 \leq z \leq \floor{m/2}$, one can expand $P$ to a polygon $Q$, such that
    \begin{compactenumi}
        \item $Q$ covers at least $z$ more grid points, and
        \item $\perim(Q) \leq \perim(P) + \frac{1}{d^3(m/4-1)}$.
    \end{compactenumi}
\end{lemma}

\begin{proof}
    Since $u$ is primitive, there is a vector $w\in\ZZ^2$ with
    $|\det(u,w)|=1$.  Thus the lattice lines parallel to $u$ occur at
    perpendicular spacing $1/d$, and each such line contains a coset
    $w+\ZZ u$ whose consecutive grid points are spaced by distance~$d$.
    Let $h$ be the adjacent lattice line parallel to $e$ on the exterior
    side of $P$, and put $H=h\cap\ZZ^2$.

    Choose two points $p$ and $q$ from $H$ closest to the middle of $e$
    such that the segment $pq$ contains exactly $z$ points of~$H$; that
    is, it has $z-2$ points of $H$ between its endpoints. Let
    $e=p'q'$, and consider the trapezoid $T=p'pqq'$, see \figref{trap}.
    \begin{figure}[ht]
        \centering {\includegraphics[page=2]{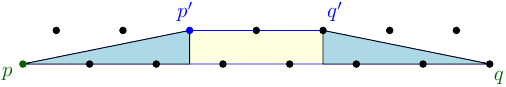}}
        \caption{}
        \figlab{trap}
    \end{figure}

    Since $z\leq m/2$, the projections of $p$ and $q$ onto $e$ are each
    at distance at least $(m/4-1)d$ from the nearer endpoint of~$e$.
    The trapezoid contains the $z$ exterior grid points on $pq$, so
    attaching $T$ to $P$ and taking the convex hull gives a new polygon
    $Q$ covering at least $z$ more grid points.  Its perimeter is larger
    by at most
    \begin{equation*}
        2 \frac{(1/d)^2}{2(m/4 -1)d}
        =
        \frac{1}{d^3(m/4-1)},
    \end{equation*}
    by using \clmref{triangle} on the two side triangles of this trapezoid.
\end{proof}

Conversely, the next lemma tells us how much perimeter we ``save'' if we cut a vertex $x$ from the optimal polygon, assuming it has two adjacent edges of length at most $\tau$.

\begin{lemma}
    \lemlab{shortcut}%
    Let $x$ be a vertex of a grid convex polygon $P$, with both edges adjacent to it having length at most $\tau$.  Then removing $x$, the residual polygon $Q = \CHX{ P \cap \ZZ^2 \setminus \{x\}}$ has perimeter $\perim(Q) \leq \perim(P) - \frac{1}{4\tau^3} = \perim(P) - \Omega(1/\tau^3)$.
\end{lemma}

\begin{proof}
    Let $p$ and $q$ be the boundary grid points immediately before and after $x$ on $\partial P$.  Instead of working with the whole triangle $\triangle pxq$, we first cut off the smallest lattice-visible tip at $x$.

    Namely, move the line through $pq$ parallel to itself toward $x$
    until it reaches the lattice line parallel to $pq$ that is adjacent
    to the line through $x$ on the $pq$ side.
    Let this line meet $xp$ and $xq$ at $p'$ and $q'$, respectively.  The
    open region between this line and $x$ contains no grid point other
    than $x$: if $u$ is the primitive direction parallel to $pq$, then
    lattice lines parallel to $u$ are separated by distance $1/\lenX{u}$,
    and this is the closest such line to the line through~$x$ on the
    $pq$ side.  Since
    $\lenX{u}\leq\lenX{pq}$, the distance from $x$ to this first line is
    at least $1/\lenX{pq}$.

    Let $Q=\CHX{P\cap\ZZ^2\setminus\{x\}}$.  Since all remaining grid points of $P$ lie on or beyond the line $p'q'$, we have
    \[
        Q\subseteq P\cap H,
    \]
    where $H$ is the halfplane bounded by the line $p'q'$ that does not contain $x$.  Therefore, by monotonicity of perimeter for convex sets,
    \[
        \perim(Q)\leq \perim(P\cap H).
    \]
    Passing from $P$ to $P\cap H$ replaces the boundary path $p'\to x\to q'$ by the segment $p'q'$.  Hence
    \[
        \perim(P)-\perim(Q) \geq \lenX{xp'}+\lenX{xq'}-\lenX{p'q'}.
    \]

    It remains to lower-bound this last quantity.  Let $s$ be the distance from $x$ to the line $p'q'$.  By the preceding paragraph, $s\geq 1/\lenX{pq}$.  Also
    \[
        \lenX{pq}\leq \lenX{px}+\lenX{xq}\leq 2\tau,
    \]
    so
    \[
        s\geq \frac{1}{2\tau}.
    \]

    \begin{figure}[ht]
        \centering {\includegraphics[page=3]{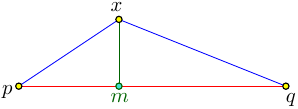}}
        \caption{}
        \figlab{drop_3}
    \end{figure}

    Let $m$ be the projection of $x$ to the line $p'q'$, see \figref{drop_3}.  We have
    \[
        \lenX{p'q'}\leq \lenX{p'm}+\lenX{mq'}.
    \]
    Therefore
    \begin{align*}
        \lenX{xp'}+\lenX{xq'}-\lenX{p'q'}
        &\geq
        (\lenX{xp'}-\lenX{p'm})
        +
        (\lenX{xq'}-\lenX{mq'}).
    \end{align*}
    Applying \clmref{triangle} to the two right triangles $\triangle p'mx$ and $\triangle q'mx$ gives
    \begin{align*}
        \lenX{xp'}+\lenX{xq'}-\lenX{p'q'}
        &\geq
        \frac{s^2}{2\lenX{xp'}}
        +
        \frac{s^2}{2\lenX{xq'}} .
    \end{align*}
    Since $p'$ lies on $xp$ and $q'$ lies on $xq$,
    \[
        \lenX{xp'}\leq \lenX{xp}\leq \tau, \qquad \lenX{xq'}\leq \lenX{xq}\leq \tau.
    \]
    Hence
    \[
        \lenX{xp'}+\lenX{xq'}-\lenX{p'q'} \geq \frac{s^2}{2\tau} + \frac{s^2}{2\tau} = \frac{s^2}{\tau} \geq \frac{1}{4\tau^3}.
    \]

    Thus
    \[
        \perim(Q)\leq \perim(P)-\frac{1}{4\tau^3},
    \]
    as claimed.
\end{proof}

Combining both results, the next lemma uses an exchange argument to upper-bound the maximum multiplicity of a primitive vector in terms of its length.

\begin{lemma}
    \lemlab{exchange}%
    Let $e$ be an edge of $\Popt$ whose generator has length $d$, and let $m\geq 8$ be the multiplicity of the generator in $e$.  Set $\tau = 3\lOpt/\nv$ and
    \begin{equation*}
        z = \min\!\left\{\floor{m/2},
          \left\lfloor \frac{\nv}{12}\right\rfloor \right\}.
    \end{equation*}
    Then
    \begin{math}
        mz = O\!\left({\tau^3}/{d^3}\right).
    \end{math}
    Consequently, if $b=O(k^{\alpha/2})$, then
    \begin{equation*}
        m =
        \begin{cases}
          O(k^{3\alpha/8}/d^{3/2}), & \text{if } m\leq \nv/6,\\
          O(k^{\alpha-1/2}/d^3), & \text{if } m>\nv/6.
        \end{cases}
    \end{equation*}
\end{lemma}

\begin{proof}
    At most $\nv/3$ edges of $\Popt$ have length larger than $\tau$.  Since each such edge has two endpoints, at most $2\nv/3$ vertices are incident to a long edge.  Thus at least $\nv/3-2$ vertices are not endpoints of $e$ and have both incident edges of length at most $\tau$. These vertices are \emphw{short}.  From these vertices we can choose at least $\floor{\nv/12}$ pairwise nonadjacent ones.

    Choose $z$ of these pairwise nonadjacent short vertices.  For each
    chosen vertex, use the local cutting halfplane from
    \lemref{shortcut}.  Since no two chosen vertices are adjacent, the
    corresponding boundary arcs are disjoint; the local shortcuts can be
    performed simultaneously, and their perimeter savings add.  Taking the
    convex hull of the remaining grid points can only decrease the
    perimeter further.  None of the chosen vertices is an endpoint of~$e$.
    All remaining grid points stay in the same closed halfplane bounded by
    the supporting line of~$e$, and the two endpoints of~$e$ remain on
    that line; hence the same segment~$e$ remains a boundary edge of the
    new hull.  Thus we obtain a polygon $P'$ whose perimeter decreased by
    $\Delta^-$,
    where
    \begin{equation*}
        \Delta^- \geq \frac{z}{4\tau^3},
    \end{equation*}
    by \lemref{shortcut}.  Since $z \leq \floor{m/2}$, \lemref{add} lets us add at least $z$ grid points along $e$ at an additional perimeter cost $\Delta^+$, where
    \begin{math}
        \displaystyle \Delta^+ \leq \frac{2}{d^3(m/4-1)}.
    \end{math}

    The resulting polygon contains at least as many grid points as
    $\Popt$.  By \obsref{trim-to-k}, it contains an exactly feasible
    $k$-point polygon with no larger perimeter.  The optimality of
    $\Popt$ therefore implies $\Delta^+ - \Delta^- \geq 0$, and
    \begin{equation*}
        \frac{2}{d^3(m/4-1)}
        \geq
        \Delta^+
        \geq
        \Delta^-
        \geq
        \frac{z}{4\tau^3}.
    \end{equation*}
    This implies $mz=O(\tau^3/d^3)$.

    If $b=O(k^{\alpha/2})$, then \lemref{vertex-lb} gives $\nv\geq \sqrt{k/b}=\Omega(k^{1/2-\alpha/4})$, and $\tau=3\lOpt/\nv=O(k^{\alpha/4})$.  If $m\leq \nv/6$, then $z=\Theta(m)$, giving $m=O(k^{3\alpha/8}/d^{3/2})$.  If $m>\nv/6$, then $z=\Theta(\nv)$, giving $m=O(k^{\alpha-1/2}/d^3)$.
\end{proof}

Finally, bootstrapping these multiplicity bounds is sufficient to yield the tight bounds on $b$ and $\WD$, and a preliminary bound on the longest edge.

\begin{theorem}
    \thmlab{improved-gap}%
    We have $\bb = O(k^{1/3})$ and $\WD = O(k^{1/6})$.  Consequently,
    \propref{edge-bound} gives the preliminary longest-edge bound
    $O(k^{1/3})$.
\end{theorem}

\begin{proof}
    We bootstrap \lemref{exchange} with \thmref{Andrews} and \propref{gap-from-b}.  Suppose $b = O(k^{\alpha/2})$ for some $\alpha \leq 6/7$.  We derive $b = O(k^{\alpha'/2})$ for a strictly smaller $\alpha'$.

    By \lemref{vertex-lb}, $\nv = \Omega(k^{1/2-\alpha/4})$.  Moreover
    \[
        \tau = \frac{3\lOpt}{\nv} = O(k^{\alpha/4}).
    \]
    By \thmref{Andrews}, $\Popt$ has at most $\nv = O(\lOpt^{2/3}) = O(k^{1/3})$ distinct edge directions.  Enumerating these in nondecreasing order of primitive length, the $j$\th vector has length $\rho_j = \Omega(\sqrt{j})$.

    Let $m_j$ be the multiplicity of the edge with primitive length $\rho_j$.  Multiplicities smaller than $8$ contribute $O(k^{1/3})$, since there are only $O(k^{1/3})$ directions, so we only consider directions with $m_j\geq 8$.

    Call a direction \emph{small} if $m_j\leq \nv/6$, and \emph{large} otherwise.  For a small direction, \lemref{exchange} gives
    \begin{equation*}
        m_j
        =
        O\pth{\frac{k^{3\alpha/8}}{\rho_j^{3/2}} }
        =
        O\pth{\frac{k^{3\alpha/8}}{j^{3/4}}}.
    \end{equation*}
    Hence the total contribution of the small directions is
    \[
        b_2%
        =%
        O\Bigl( k^{3\alpha/8}\sum_{j=1}^{O(k^{1/3})} j^{-3/4} \Bigr)%
        =%
        O(k^{3\alpha/8+1/12}).
    \]

    For a large direction, \lemref{exchange} gives $m_j=O(k^{\alpha-1/2}/\rho_j^3)$.  Since also $m_j>\nv/6=\Omega(k^{1/2-\alpha/4})$, a large direction must satisfy
    \[
        \Omega(k^{1/2-\alpha/4}) = m_j = O(k^{\alpha-1/2}/\rho_j^3) \implies \rho_j = O(k^{(5\alpha/4-1)/3}).
    \]
    Call this threshold $\rho_*$.  If $\alpha<4/5$, then $\rho_*<1$, so no large direction exists.  If $\alpha\geq4/5$, then there are $O(\rho_*^2)$ large directions.  For them we use the maximum-edge-length bound $m_j\rho_j\leq\emax$.  Here $\WD=O(k^{\alpha/4})$ by \propref{gap-from-b}, so
    \begin{equation*}
        \emax=O(\sqrt{\WD} {k}^{1/4 })=O(k^{1/4+\alpha/8})
    \end{equation*}
    by \propref{edge-bound}.  Thus
    \[
        b_1 \leq \emax\sum_{j=1}^{O(\rho_*^2)}\frac1{\rho_j} = O(\emax\rho_*) = O(k^{13\alpha/24-1/12}).
    \]
    Since $13\alpha/24-1/12\leq 3\alpha/8+1/12$ for $\alpha\leq1$, the large directions are dominated by $b_2$.

    Therefore $b = O(k^{3\alpha/8 + 1/12})$.  Setting $b = O(k^{\alpha'/2})$ gives the recurrence $\alpha' = \tfrac{3}{4}\alpha + \tfrac{1}{6}$.  Starting from $\alpha_0 = 6/7$:
    \begin{equation*}
        \alpha_1
        = \tfrac{3}{4} \cdot \tfrac{6}{7} + \tfrac{1}{6}
        = \tfrac{17}{21}
        \approx 0.810,
        \qquad
        \alpha_2
        = \tfrac{3}{4} \cdot \tfrac{17}{21} + \tfrac{1}{6}
        = \tfrac{65}{84}
        \approx 0.774
        <
        \tfrac{4}{5},
    \end{equation*}
    so after two bootstrap steps we have $\alpha<4/5$, and hence no large directions.  At this point we close the estimate directly.  Since $\tau = O(\lOpt/\nv)=O(\sqrt{k}/\sqrt{k/b})=O(\sqrt b)$, \lemref{exchange} gives $m_j=O(b^{3/4}/j^{3/4})$ for every direction with $m_j\geq8$.  Therefore
    \begin{equation*}
        b
        \leq
        O(k^{1/3})
        +
        O\Bigl( b^{3/4}\sum_{j=1}^{O(k^{1/3})} j^{-3/4}\Bigr)
        =
        O\pth{b^{3/4} k^{1/12}} + O(k^{1/3}),
    \end{equation*}
    which implies $b = O(k^{1/3})$.  Then $\WD = O(\sqrt{\bb}) = O(k^{1/6})$ by \propref{gap-from-b}, and \propref{edge-bound} gives the preliminary longest-edge bound $O(\sqrt{\WD} k^{1/4}) = O(k^{1/3})$.
\end{proof}

\begin{corollary}
    \corlab{nv_bound}%
    Using the notations of \tabref{notations}, we have $\nv = \Omega(k^{1/3})$.
\end{corollary}
\begin{proof}
    By \thmref{improved-gap}, we have $\bb = O(k^{1/3})$. By \lemref{vertex-lb}, $\nv \geq \sqrt{k/b} = \Omega(k^{1/3})$.
\end{proof}

Using \clmref{lower_1} with the improved bounds of \thmref{improved-gap} implies the following.

\begin{corollary}
    \corlab{perimeter-window}
    We have $ 2\sqrt{\pi k}-O(k^{-1/6}) \leq \lOpt \leq 2\sqrt{\pi k}.  $
\end{corollary}

To obtain the final $O(k^{1/4})$ bound on the longest edge, we require a slightly more refined version of the argument above.

\begin{lemma}
    \lemlab{add_vertex}%
    Let $P$ be a convex grid polygon, and let $u$ be a primitive
    direction utilized by an edge of $P$.
    Then there is a convex grid polygon covering at least
    $\cardin{P\cap\ZZ^2}+1$ grid points, with perimeter
    $\leq \perimX{P}+O(1/\lenX{u}^2)$.
\end{lemma}

\begin{proof}
    Let $d=\lenX{u}$, and let $AB$ be an edge of $P$ whose
    primitive direction is $u$.  If $d=O(1)$, choose a primitive
    subsegment $[p,p+u]$ of $AB$.  Since $u$ is primitive, there is a
    grid point on each adjacent lattice line forming a primitive triangle
    with this subsegment; one of the two is on the exterior side of
    $P$.  Adding that point and taking the convex hull increases the
    perimeter by at most a constant, which is $O(1/d^2)$ in this case.
    Thus assume $d$ is larger than an absolute constant.

    After reflecting and swapping the coordinate axes if necessary,
    write $u=(a,b)$ with $1\leq a\leq b$ and $\gcd(a,b)=1$.  In the
    remaining nontrivial case $b\geq2$.  Choose
    $i\in\{1,\ldots,b-1\}$ and $x\in\ZZ$ with
    $ai-bx=1$, and put $w=(x,i)$.  Then
    $\det(u,w)=1$, so $w$ lies on a lattice line parallel to the line
    spanned by $u$ and at perpendicular distance $1/d$ from it.

    Let $\alpha$ be the distance, measured along direction $u$, from the
    origin to the projection of $w$ onto $\Re u$.  Since
    $x=(ai-1)/b$ and $d^2=a^2+b^2$,
    \[
        \alpha
        =
        \frac{w\cdot u}{d}
        =
        \frac{ax+bi}{d}
        =
        \frac{i d}{b}-\frac{a}{bd}.
    \]
    Hence
    \[
        \alpha\geq \frac{d}{b}-\frac{1}{d}=\Omega(1),
        \qquad
        d-\alpha
        =
        \frac{(b-i)d}{b}+\frac{a}{bd}
        =
        \Omega(1).
    \]
    Thus this adjacent lattice point projects to the interior of the
    primitive segment $[0,u]$, at distance bounded below by an absolute
    constant from both endpoints.

    Place the construction on a primitive subsegment $[p,p+u]$ of
    $AB$.  The two points $p+w$ and $p+u-w$ lie on the two adjacent
    lattice lines.  One of them is on the exterior side of $P$; call
    it $w'$, and let $r$ be its orthogonal projection onto $AB$.  The
    point $w'$ is a new grid point, its height above $AB$ is $1/d$, and
    $r$ is at distance $\Omega(1)$ from both ends of the chosen
    primitive subsegment.

    Let $L=\lenX{AB}$ and $t=\lenX{Ar}$.  The polygonal curve obtained
    from $\partial P$ by replacing the edge $AB$ with the two
    segments $Aw'$ and $w'B$ contains
    $\CHX{P\cup\{w'\}}$, and hence upper-bounds its perimeter.
    Since $t$ and $L-t$ are both $\Omega(1)$, \clmref{triangle} gives
    \[
        \perimX{\CHX{P\cup\{w'\}}}-\perimX{P}
        \leq
        \frac{1}{d^2t}+\frac{1}{d^2(L-t)}
        =
        O(1/d^2).
    \]
    The hull contains all grid points of $P$ and also the exterior
    grid point $w'$, so it covers at least
    $\cardin{P\cap\ZZ^2}+1$ grid points.
\end{proof}

\begin{lemma}
    \lemlab{p_d_bound}%
    Every primitive direction that appears on an edge of $\Popt$ has length $O(k^{1/4})$.
\end{lemma}

\begin{proof}
    Let $e$ be an edge of $\Popt$ whose primitive direction $u$ is
    longest among all primitive edge directions, and put $d=\lenX{u}$.

    We pay for the future insertion along $e$ by first deleting a cheap
    vertex elsewhere.  By \thmref{improved-gap}, $b=O(k^{1/3})$, and therefore \lemref{vertex-lb} gives $\nv=\Omega(k^{1/3})$.  Also $\lOpt=O(\sqrt{k})$, so
    \[
        \tau=\frac{3\lOpt}{\nv}=O(k^{1/6}).
    \]
    At most $\nv/3$ edges have length larger than $\tau$, and hence at least $\nv/3$ vertices have both incident edges of length at most $\tau$. Since at least $\nv/3 > 2$ (for sufficiently large $k$) vertices are short, there is a short vertex $x'$ not incident to $e$. By \lemref{shortcut}, deleting $x'$ and taking the convex hull decreases the perimeter by
    \[
        \Omega(1/\tau^3)=\Omega(k^{-1/2}).
    \]

    Perform the deletion first.  Since $x'$ is not incident to~$e$, both
    endpoints of~$e$ remain, and all remaining grid points are still in
    the same closed halfplane bounded by the supporting line of~$e$.
    Therefore the same segment~$e$ is still a boundary edge of the new
    hull.  Then apply \lemref{add_vertex} on that edge, at cost
    $O(1/d^2)$.  The combined
    operation deletes exactly one old grid point and adds at least one
    new grid point, so the resulting polygon contains at least as many
    grid points as $\Popt$.  By \obsref{trim-to-k} and the optimality of
    $\Popt$, this forces
    \begin{math}
        1/\sqrt{k} = O(1/d^2),
    \end{math}
    and $d=O(k^{1/4})$.
\end{proof}

\begin{corollary}
    \corlab{edge-improved}%
    The longest edge of $\Popt$ has length $\emax=O(k^{1/4})$.
\end{corollary}

\begin{proof}
    Let $e=m u$ be an edge of $\Popt$, where $u$ is primitive and $d=\lenX{u}$.  If $m\leq 8$, then $\lenX{e}=md=O(k^{1/4})$ by \lemref{p_d_bound}.  Assume then that $m>8$. This is the only case in which multiplicity might make the edge longer than its primitive direction.

    We now apply \lemref{exchange} with $\alpha=2/3$, which is allowed because \thmref{improved-gap} gives $b=O(k^{1/3}) =O(k^{\alpha/2})$.  By \lemref{exchange}, with $\alpha=2/3$,
    \[
        \lenX{e} =%
        md =%
        \begin{cases}
          O(k^{1/4}/d^{1/2}), & \text{if } m\leq \nv/6,\\
          O(k^{1/6}/d^2), & \text{if } m>\nv/6.
        \end{cases}.
    \]
    Since $d\geq1$, the first case is $O(k^{1/4})$ and the second is $O(k^{1/6})\subseteq O(k^{1/4})$.
\end{proof}

\subsubsection{Bounding the turn angle}

The next property of $\Popt$ to bound is the turning angel.

\begin{defn}
    Consider a convex polygon $P$. As one traverses its boundary in the counterclockwise direction, it enters the $i$\th vertex on a direction $d_i$ and leaves on a direction $d_{i+1}$. The angle between these two directions, is the \emphi{turn angle} at the $i$\th vertex. The maximum such quantity is the \emphi{turn angle} of $P$.
\end{defn}

\begin{lemma}
    \lemlab{turn_angle_2} The maximum turn angle of $\Popt$ is $O(1/k^{1/6})$.
\end{lemma}
\begin{proof}
    Let $c$ be the epicenter of the minimum-width concentric disk
    sandwich for $\Popt$, and let its radii be $\rin$ and $\Rout$.  By
    \thmref{improved-gap}, $\Rout-\rin=O(k^{1/6})$.  Also
    $\Rout=\Omega(\sqrt{k})$, since
    $\Popt\subseteq \diskY{c}{\Rout}$ and Pick's theorem gives
    $\areaOpt=k-O(k^{1/3})$.  Hence $\rin=\Omega(\sqrt{k})$.

    Consider a vertex $v$ with turn angle $\alpha$, and let $\ell_1$ and
    $\ell_2$ be the supporting lines of the two incident edges.  Since
    $\diskY{c}{\rin}\subseteq \Popt$, each of these supporting lines is
    at distance at least $\rin$ from~$c$.  The angle between the two
    outward normals is exactly~$\alpha$.  Among two lines with this angle
    and distances at least $\rin$ from~$c$, their intersection is closest
    to~$c$ when both distances are $\rin$ and the configuration is
    symmetric; in that case the distance is $\rin/\cos(\alpha/2)$.
    Therefore
    \[
        \Rout \geq \dY{c}{v} \geq \frac{\rin}{\cos(\alpha/2)}.
    \]
    Thus
    \[
        1-\cos(\alpha/2)
        \leq
        1-\frac{\rin}{\Rout}
        =
        \frac{\Rout-\rin}{\Rout}
        =
        O(k^{-1/3}).
    \]
    Since $1-\cos x=\Omega(x^2)$ for $x\in[0,\pi/2]$, it follows that
    $\alpha=O(k^{-1/6})$.
\end{proof}

\section{The dynamic programming algorithm}
\seclab{dp}

We present below our \DP algorithm for computing the optimal polygon $\Popt$, given $k$. The basic idea is to construct a visibility graph over a set of grid points $\VV$ that can serve as boundary grid points on $\partial \Popt$.  The dynamic program itself ``walks'' around $\Popt$, constructing it (or, more precisely, constructing various competing partial solutions that might end up forming $\Popt$). Importantly, the \DP step is between two boundary points (not vertices!) connected via a primitive vector.

In the following, we show a sequence of \DP algorithms with escalating sophistication improving the running times of the \DP as we go along. In particular, we show:
\begin{compactenumI}
    \item \textsc{Plug in:} Using the bounds from the previous section into the natural \DP, yields a \DP with running time $O(k^{2+1/2})$, see \thmref{dp-general} below.

    \item \textsc{Limiting the search space to a ring:} The above analysis implies that the optimal solution must lie in a pretty tight ring. By explicitly computing this ring-like area, we can reduce the \DP space significantly. This improves the running time to $O(k^{2+1/6})$, see \lemref{dp-annulus}.

    \item \textsc{Cutting into thinner rings:} The search can be restricted to a small number of thinner rings. This improves the running time to $O(k^2)$, see \lemref{dp-annulus_3}.

    \item \textsc{Smaller turns:} Using that the optimal solution has a small turn angle, we can further limit the number of transitions the algorithm needs to consider. This improves the running time to $O(k^{2-1/6})$, see \lemref{dp-annulus_4}.

    \item \secref{rings_2} carefully argues that, after choosing a suitable grid point on a supporting edge as the anchor, the epicenter has small horizontal uncertainty. This enables us to restrict the search space for candidate centers even further, resulting in an improved running time $O(k^{2-2/9})$, see \lemref{dp-annulus_5}.

    \item \textsc{A final technical speedup:} The last step batches the
    \DP transitions between columns using an underlying monotonicity
    property, see \secref{monotonicity}.  This is independent of the
    geometric center-location reduction, and combined with it gives the
    final running time $O(k^{29/18+o(1)})$.
\end{compactenumI}

\subsection{Basic tools}

A portion of $\partial \Popt$ between two consecutive grid points is a \emphi{link}. A link is always formed by a primitive vector. Several links with the same direction might form an edge of $\Popt$.

\paragraph{Edge-vector encoding.}
By \lemref{tight-upper}, $\UBlOpt = 2\sqrt{\pi k}$ is an upper bound on the optimal perimeter.  Fix also an upper bound $\emax = O(k^{1/4})$ on the maximum edge length of $\Popt$ (\corref{edge-improved}), and hence on the length of every boundary link used by the \DP.

The first step is to translate the polygon so that a chosen boundary grid
point on a supporting line is at the origin and, after a rotation if
needed, the polygon lies in the upper halfplane.  In the basic runs this
anchor can be the bottom-right vertex.  In the final center search
(\thmref{dp-leftmost}) the anchor may be an interior grid point of a
supporting edge, which is why the \DP is formulated in terms of boundary
grid points rather than only vertices.  An edge of $\Popt$ can be viewed
as a directed vector, as we traverse $\partial \Popt$, say, in the
counterclockwise direction.  Let $\npvC = \npvX{\emax}$ be the number of
oriented primitive grid vectors of length at most $\emax$, and order them as
\begin{equation*}
    u_1, u_2, \dots, u_\npvC
\end{equation*}
in strictly increasing angular order.  Let $\rho_j = \lenX{u_j}$ denote the length of the $j$\th primitive vector in this list.

\begin{observation}
    \obslab{direction-count}%
    The total number of primitive vectors of length at most $\emax$ is $\Theta(\emax^2)$, as this is bounded by the number of grid points in the square $[-\emax,\emax]^2$. The lower bound is implied by the known bounds on the totient summatory function \cite{hw-itn-08}.
\end{observation}

\paragraph*{Good points.}
Let $k' = 2\ceil{\smash{\sqrt{k}}}$.  All relevant vertices lie in the set $\Good = \IRY{-k'}{k'} \times \IRY{0}{k'}$. Observe that $\cardin{ \Good}=O(k)$.

\paragraph*{Counting grid points on an edge.}
The number of grid points internal to the edge $pq$, denoted by $\#_{pq}$, with $p,q \in \Good $, can be computed by translating the segment so that it starts at the origin and ends at $z= q-p$.  This quantity is $(\gcdY{z_x}{z_y}) - 1 $, and it can be precomputed and queried in constant time.

\paragraph*{Counting grid points in a triangle.}
The number of grid points covered by a triangle $\triangle = \triangle \origin p q$, denoted by $\#_{\triangle}$, with $p,q \in \Good $, can be computed by using \thmrefY{pick}{Pick's theorem}. Indeed, the number of grid points, denoted by $\#_\triangle$, contained in $\triangle$ is
\begin{equation*}
    \#_\triangle
    =
    \areaC_\triangle
    +  b_\triangle /2 + 1  .
\end{equation*}

\paragraph{The \DP state and transitions.}
Order the points of $\Good\setminus\{\origin\}$ by the polar angle of the
ray from~$\origin$, breaking ties along a common ray by distance from
$\origin$.  A partial solution is a polygonal chain
\[
    \origin=p_0,p_1,\ldots,p_t=s
\]
whose points are increasing in this order.  Together with the segment
$s\origin$, this chain bounds the fan
$\bigcup_{i=1}^{t-1}\triangle\origin p_i p_{i+1}$; the triangles have
disjoint interiors because the endpoints are in polar order.  The fan need
not be convex; this is harmless.  If a completed fan contains $k$ grid
points, then the convex hull of the fan contains at least $k$ grid points
and has perimeter no larger than the chain.  By \obsref{trim-to-k}, it
contains a convex grid polygon with exactly $k$ grid points and no larger
perimeter.
Conversely, if a feasible convex polygon has $\origin$ on its boundary,
then every ray from~$\origin$ intersects it in a segment, so the boundary
arc from the first edge after~$\origin$ to the last edge before returning
to~$\origin$ is encountered in nondecreasing polar order.  We keep the two
boundary portions incident to~$\origin$ as the source and sink
transitions, and refine every other boundary portion into primitive
links.  This includes the case where $\origin$ lies in the relative
interior of a polygon edge, so the two incident portions are the two
subsegments of that edge.  This forms one of the chains considered by the
\DP.  Thus minimizing over these fan chains gives exactly the same
optimum as the original convex problem.

A state of the \DP is defined by a pair $(s, \nG)$:
\begin{compactenumI}
    \item $s \in \Good$: The state encodes a chain starting at the origin and ending at $s$, processed in counterclockwise angular order around the origin.

    \item $\nG$: The convex polygon constructed so far covers $\nG$ grid points.
\end{compactenumI}
For each state, the \DP stores the minimum perimeter among partial solutions with the prescribed parameters.  The start state is $( \origin, 1)$, where $\origin$ denotes the origin.

For an internal transition from $(s, \nG)$, use a primitive vector $u_j$,
for $j \in \IRX{\npvC}$, and set $s' = s + u_j$.  The \DP keeps the
transition only if $s'\in\Good$, $s'$ is later or equal to $s$ in the
above polar order.
The source
and sink transitions are allowed to use any grid vector of length at
most~$\emax$, since these represent the two boundary portions incident
to~$\origin$.  Even when $\origin$ lies in the relative interior of an
edge, these two portions are subsegments of that edge, and hence have
length at most~$\emax$.  The triangle
$\triangle = \triangle \origin s s' $ is added to the constructed
polygon.  If the state $(s, \nG)$ had perimeter $\tau$, then the new
polygon has perimeter
\begin{equation*}
    \tau' = \tau - \dY{\origin}{s} + \dY{\origin}{ s'} + \dY{s}{ s'}.
\end{equation*}
The number of grid points covered by the new polygon is
\begin{equation*}
    \nG' = \nG + \#_\triangle - \#_{\origin s} - 2.
\end{equation*}
The subtraction removes the shared diagonal $\origin s$, including its two endpoints.  The formula applies to the nondegenerate fan triangles used by internal transitions.  A degenerate transition along a ray from~$\origin$ only adds the previously uncounted grid points on that ray; this is used only for the source edge and, symmetrically, the sink edge adds no grid points and only removes the last closing diagonal.  Thus a transition is created from the current state $(s,\nG)$, with current perimeter $\tau$, to a new state $( s', \nG')$, with $\tau'$ being the suggested new perimeter.  If this state already has an associated solution with a shorter perimeter, the \DP does not update the stored value.

The \DP explores the states in angular order, with $\nG$ ranging from $0$ to $k$.  The optimal solution is the minimum perimeter of any completed state whose $\nG$ value is $k$.

\subsubsection{Analysis}

\begin{theorem}
    \thmlab{dp-general}%
    Given an upper bound $\emax$ on the maximum length of links, the \DP computes the optimal polygon in $O(k^2 \emax^2)$ time. Consequently, the \DP runs in $O(k^{2+1/2})$ time.
\end{theorem}

\begin{proof}
    The number of states is $\cardin{\Good} (k+1) = O(k^2)$.  Every internal state has $\npvC = O(\emax^2)$ outgoing primitive steps by \obsref{direction-count}; the source and sink choices add only $O(\emax^2)$ possible grid vectors of length at most~$\emax$.  It follows that the running time is $O(k^2 \emax^2)$.  Using $\emax = O(k^{1/4})$ (\corref{edge-improved}), the basic \DP takes $O(k^{5/2})$ time.
\end{proof}

\subsection{Constraining the solution via an annulus}

\begin{figure}[t]
    \centering \hfill%
    \includegraphics[page=1]{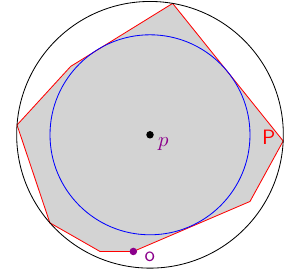} \hfill%
    \includegraphics[page=2]{figs/around} \hfill%
    \phantom{}%
    \caption{In search of a good center.}
    \figlab{chord}
\end{figure}

\begin{lemma}
    \lemlab{in}%
    After translating an optimal polygon so that a boundary grid point on
    a bottom supporting line is the origin and the polygon lies in the
    halfplane $y\geq0$, one can \emph{explicitly} compute a point $q$ and
    radii $r<R$ such that every such translate of $\Popt$: (i) contains
    $\Din=\diskY{q}{r}$, and (ii) it is contained in $\Dout=\diskY{q}{R}$,
    where $R-r=O(k^{1/3})$ and $R=\Theta(\sqrt{k})$.  Moreover, the
    epicenter lies in the explicit rectangle
    \[
        \cardin{x}\leq C k^{1/3},
        \qquad
        \cardin{y-\sqrt{k/\pi}}\leq C k^{1/6},
    \]
    for a sufficiently large absolute constant~$C$.
\end{lemma}
\begin{proof}
    Let $\diskY{p}{\rin}\subseteq\Popt\subseteq\diskY{p}{\Rout}$ be the
    concentric disk sandwich from \thmref{improved-gap}, with
    $\WD = \Rout - \rin = O(k^{1/6})$.  Put
    $\rho=\sqrt{k/\pi}$.  Since
    $\areaOpt=k-O(k^{1/3})$, the inclusions
    $\diskY{p}{\rin}\subseteq\Popt\subseteq\diskY{p}{\Rout}$ imply
    \[
        \rin,\Rout=\rho+O(k^{1/6}).
    \]
    Indeed, $\pi\rin^2\leq\areaOpt\leq\pi\Rout^2$ and
    $\Rout-\rin=O(k^{1/6})$.

    After translating so that the chosen support grid point is the
    origin, $\Popt$ lies in the halfplane $y\geq 0$.  Since the inner disk
    lies inside $\Popt$, we have $p_y\geq \rin$, see \figref{chord}.  Also
    the origin is on $\partial\Popt$, so
    $\rin\leq\dY{p}{\origin}\leq\Rout$.  Thus, $p$ lies in the sector
    \begin{equation*}
        S = \pth{ \bigl. \diskY{\origin}{\Rout} \setminus
           \diskY{\origin}{\rin} } \cap ( y \geq \rin ).
    \end{equation*}
    The middle point of the bottom chord of this sector is the point
    $s=(0,\rin)$.  The distance $\dY{s}{p}$ is bounded by the diameter
    of $S$.  This diameter is
    $O(\sqrt{\Rout(\Rout-\rin)})=O(k^{1/3})$, by the same argument used
    in \propref{edge-bound}.
    Also $p_y\in[\rin,\Rout]$ and
    $\rin,\Rout=\rho+O(k^{1/6})$.  Thus $p$ lies in the stated explicit
    rectangle after increasing the constant~$C$.

    Set $q=(0,\rho)$.  Since $\rin=\rho+O(k^{1/6})$, we have
    $\dY{s}{q}=O(k^{1/6})$, and by the triangle inequality
    $\dY{p}{q}=O(k^{1/3})$.  Choose a sufficiently large absolute
    constant $C$ and set
    \begin{equation*}
        r=\rho-Ck^{1/3}
        \qquad\text{and}\qquad
        R=\rho +Ck^{1/3}.
    \end{equation*}
    Increasing $C$ beyond the hidden constants gives
    $\dY{p}{q}+r<\rin$ and $\Rout+\dY{p}{q}<R$.  Hence
    $\Din = \diskY{q}{r}$ is contained in the interior of
    $\diskY{p}{\rin}$, and therefore in the interior of~$\Popt$, while
    $\Popt\subseteq\diskY{p}{\Rout}\subseteq \Dout = \diskY{q}{R}$.
\end{proof}

We can now use $\Din$ and $\Dout$ to speed up our search algorithm.  All the points in $B' = \Din \cap \ZZ^2 \cap (y \geq 0) $ are ``bad'', in the sense that the optimal solution contains them in its interior and therefore cannot use them as boundary points. Similarly, the points of $\VV' = (\Dout \setminus \Din) \cap \ZZ^2 \cap (y \geq 0)$ are ``good'' in the sense that the optimal solution might use them as vertices for $\Popt$. We can do slightly better.

\begin{defn}
    \deflab{forbidden} Let $\mathcal{B}=\CHX{B'\cup\{\origin\}}$ be the forbidden region, extending $B'$ by the shadow of the origin, and let $B=(\mathcal{B}\cap\ZZ^2)\setminus\{\origin\}$ be the corresponding forbidden grid set.  The set $\VV = (\VV' \setminus \mathcal{B})\cup\{\origin\}$ is the \emphi{good} set.
\end{defn}

Two points $p,q\in\VV$ are \emphw{visibly~compatible} if $\mathcal{B}$ is to the left of the directed line $pq$ and the line supporting $pq$ does not intersect $\mathcal{B}$ except possibly at~$\origin$.  We also require the supporting line to avoid the interior of the inner disk~$\Din$.  This tangent filter cannot remove a boundary link of $\Popt$, since $\Din\subseteq\Popt$ and every boundary link lies on a supporting line of~$\Popt$.  A visibly compatible segment may be nonprimitive.

Let $\prec$ be the polar order of $\VV\setminus\{\origin\}$ around
$\origin$, with distance from $\origin$ used to break ties on a ray.  We
create a sink copy~$\origin'$.  The graph $\G=(\VV\cup\{\origin'\},\EE)$
contains the directed edge $p\to q$ only when the following hold: $p,q$ are
visibly compatible after identifying $\origin'$ with $\origin$;
$\lenX{q-p}\leq\emax$; and either $p=\origin$, or $q=\origin'$, or
$p\prec q$.  For internal edges, where neither endpoint is an origin copy,
we additionally require $q-p$ to be primitive and the two endpoints to lie
on distinct rays from~$\origin$.  Source and sink edges may be nonprimitive
and may be radial; their boundary grid points are counted directly by the
corresponding gcd.  This makes $\G$ acyclic with the single source~$\origin$
and the single sink~$\origin'$.  Recall that $\emax=O(k^{1/4})$ by
\corref{edge-improved}, so no edge of $\Popt$ is lost.

The restricted graph preserves the optimum in the same sense as the basic
\DP.  Indeed, $B'$ is contained in the interior of~$\Popt$ and
$\origin\in\Popt$, so by convexity
$\mathcal{B}\setminus\{\origin\}$ is contained in the interior of~$\Popt$.
Hence the two boundary portions incident to~$\origin$ and every primitive
boundary link on the remaining boundary arc pass the visibility and tangent
tests, lie in~$\VV$, and appear as edges of~$\G$.  Conversely, any
source--sink path of~$\G$ defines a fan chain; if the \DP count of this fan
is $k$, then taking its convex hull and applying \obsref{trim-to-k} yields a
feasible convex grid polygon of no larger perimeter.

\begin{lemma}
    \lemlab{forced-count}%
    For every $p\in\VV\setminus\{\origin\}$, all source--$p$ paths in~$\G$
    have grid-count values in an interval of length at most
    $\cardin{\VV}+1$.  The source state at~$\origin$ has a single
    grid-count value.
\end{lemma}
\begin{proof}
    Extend the polar order $\prec$ to the grid points in
    $B\cup\VV\setminus\{\origin\}$, breaking ties on a ray by distance
    from~$\origin$.  Let
    \[
        B_{\preceq p}=\Set{z\in B}{z\preceq p},
        \qquad
        f_p=\cardin{B_{\preceq p}}.
    \]
    Let $\gamma$ be a source--$p$ path, written as
    $\origin=p_0,p_1,\ldots,p_t=p$, and let
    \[
        F(\gamma,p)=\bigcup_{r=0}^{t-1}\triangle\origin p_r p_{r+1}
    \]
    be the closed fan already swept by this partial path.  The first
    triangle is interpreted as the source radial segment, and all later
    internal triangles are nondegenerate by the graph definition.

    We first prove that every point of $B_{\preceq p}$ lies in this fan.
    Consider $z\in B_{\preceq p}$, and choose the first index~$r$ such
    that $z\preceq p_{r+1}$.  If $r=0$, then visibility of the source edge
    $\origin p_1$ implies that no point of $\mathcal{B}$ with smaller
    angular order lies strictly before this ray: $\mathcal{B}$ is on the
    swept side of the directed line $\origin p_1$, which in the upper
    halfplane is the side after the ray.  Thus only points on the ray
    $\origin p_1$ can occur in this case, and they are counted by the
    source-edge gcd term with the distance-tie rule deciding boundary
    ownership.  Otherwise, $z$ has polar order between $p_r$ and
    $p_{r+1}$.  The edge $p_rp_{r+1}$ is visibly compatible, so
    $\mathcal{B}$ is on the swept side of the directed line
    $p_rp_{r+1}$, and this line does not cut through $\mathcal{B}$ except
    possibly at~$\origin$.  Hence the part of $\mathcal{B}$ in this
    angular wedge is contained in
    $\triangle\origin p_rp_{r+1}$, and therefore $z\in F(\gamma,p)$.
    Boundary-ray points are again assigned by distance order and counted by
    the gcd term of the corresponding radial side.  Thus every state ending
    at~$p$ has grid count at least~$f_p$.

    Conversely, no point of $B\setminus B_{\preceq p}$ lies in
    $F(\gamma,p)$.  Each nondegenerate triangle
    $\triangle\origin p_rp_{r+1}$ is contained in the closed angular wedge
    between its two endpoint rays, and the source radial segment contains
    only its own ray.  Since $p_0,\ldots,p_t=p$ are increasing in the polar
    order, their union has angular support only up to~$p$, with same-ray
    boundary points assigned by the distance-tie rule.  Thus the forbidden
    points swept by a source--$p$ path are exactly $B_{\preceq p}$.

    All vertices of a path lie in $\VV\subseteq\Dout$, and $\Dout$ is
    convex, so the swept fan contains no grid point outside $\Dout$.
    Inside $\Dout$, the grid points are partitioned into the forbidden set
    $B$ and the good annulus set $\VV$ (with $\origin$ counted separately).
    Therefore any variation in the grid count among paths ending at~$p$
    comes only from grid points of~$\VV$, giving
    \[
        f_p\leq \nG\leq f_p+\cardin{\VV}.
    \]
\end{proof}

\begin{lemma}
    \lemlab{visibility} %
    The graph $\G$ has $\cardin{\VV} =O(k^{5/6})$ vertices, and $\cardin{\EE}= O(k^{4/3})$ edges, with the maximum degree in $\G$ being $\Delta = O( k^{1/2})$.
\end{lemma}

\begin{proof}
    By \lemref{in}, the set $\VV$ lies in an annulus of outer radius $R = O(\sqrt{k})$ and width $w = O(k^{1/3})$, and thus its area is $ O(w R) = O( k^{5/6})$, which also bounds the number of grid cells intersecting it, and thus the number of grid points in $\VV$, implying $\cardin{\VV} = O( k^{5/6})$.

    Every internal edge $u \to w$ in $\G$ is associated with the unique primitive grid vector $w - u$ of length at most $\emax$.  Source and sink edges are not necessarily primitive, but there are only $O(\emax^2)$ grid vectors of length at most~$\emax$.  By \obsref{direction-count}, the number of primitive grid vectors of length at most $\emax$ is also $O(\emax^2)$.  Using the refined edge bound $\emax = O(k^{1/4})$ (\corref{edge-improved}), we have
    \[
        \Delta = O(\emax^2) = O(k^{1/2}).
    \]

    Thus, the number of edges in the graph is $\cardin{\EE} \leq \cardin{\VV} \cdot \Delta = O(k^{5/6} \cdot k^{1/2}) = O(k^{4/3})$.
\end{proof}

The new \DP algorithm uses only the vertices of $\G$, starting at the origin, and only edges of $\G$ for transitions.  It precomputes, for every edge in $\G$, the corresponding grid-count and perimeter updates, so that the \DP can process an edge in constant time. We thus get the following.
\begin{lemma}
    \lemlab{dp-annulus}%
    The above dynamic program computes the smallest-perimeter polygon containing $k$ points of the grid in $O( k^{13/6})=O( k^{2+1/6})$ time.
\end{lemma}
\begin{proof}
    By \lemref{visibility}, there are $S = O( |\VV| k ) = O(k^{11/6})$ states, and each state requires $D = O( \Delta ) = O( k^{1/2})$ time to handle, where $\Delta$ is the maximum degree of $\G$.  Thus, the overall running time of the \DP is $O( S D ) = O( k^{7/3} )$.

    We can, however, do better.  By \lemref{forced-count}, each endpoint
    $p$ participates in at most $\cardin{\VV}+1$ possible grid-count
    states.  Store only these reachable sparse states, rather than a dense
    $\cardin{\VV}\times k$ table.  The number of states is
    $S=O(\cardin{\VV}^2)=O(k^{5/3})$, and the overall time is
    $O(SD)=O(k^{5/3}\cdot k^{1/2})=O(k^{13/6})$.
\end{proof}

\subsubsection{Smaller rings, faster algorithm}

A natural approach to speeding up the algorithm is to perform a tighter search for the epicenter $\epi$ of the minimum-width concentric disk sandwich, which in turn enables us to use a narrower ring in the dynamic program.

\begin{lemma}
    \lemlab{dp_annulus_2}%
    Given a parameter $\rho\geq C_0 k^{1/6}$ and a point $q$ at distance
    at most~$\rho$ from the epicenter $\epi$ (where $\Popt$ is anchored at
    a bottom support grid point at the origin), the \DP algorithm can be
    modified to run in $O(k^{3/2} \rho^2)$ time.
\end{lemma}
\begin{proof}
    By \thmref{improved-gap}, the width of the minimum-width concentric
    disk sandwich for $\Popt$ is $O(k^{1/6})$.  As in \lemref{in}, its two
    radii are $\sqrt{k/\pi}+O(k^{1/6})$. Thus, setting
    \[
        r = \sqrt{ k/\pi} - C \rho
        \qquad\text{and}\qquad
        R = \sqrt{ k/\pi} + C \rho
    \]
    with $C$ sufficiently large gives
    \[
        \diskY{q}{r}\subset \operatorname{int}(\Popt)
        \qquad\text{and}\qquad
        \Popt\subseteq\diskY{q}{R}.
    \]
    Construct the forbidden region, good set, visibility graph, and
    tangent filter exactly as in the proof of \lemref{dp-annulus}, using
    this inner disk and outer disk.  Then no boundary link of $\Popt$ is
    lost, and \lemref{forced-count} applies to this run as well.  Hence
    each endpoint has at most $\cardin{\VV}+1$ possible grid-count states.

    The good set lies in an annulus of radius $\Theta(\sqrt{k})$ and
    width $O(\rho)$, so $\cardin{\VV}=O(\sqrt{k}\rho)$.  The number of
    states is therefore $O(\cardin{\VV}^2)=O(k\rho^2)$.  By
    \lemref{p_d_bound}, every state has $O(\sqrt{k})$ transitions out of
    it. We conclude that the resulting \DP has running time
    $O(k^{3/2} \rho^2)$.
\end{proof}

\begin{lemma}
    \lemlab{dp-annulus_3}%
    The optimal polygon $\Popt$ can be computed in $O(k^2)$ time.
\end{lemma}
\begin{proof}
    \lemref{in} gives an explicit rectangle of width $O(k^{1/3})$ and
    height $O(k^{1/6})$ containing the epicenter.  Place
    $O(k^{1/6})$ candidate points along the middle horizontal line of this
    rectangle, spaced $\Theta(k^{1/6})$ apart.  Since the rectangle has
    height $O(k^{1/6})$, one candidate is at distance $O(k^{1/6})$ from
    the epicenter.  Now deploy the algorithm of \lemref{dp_annulus_2} for
    each one of these points. As $\rho=O(k^{1/6})$, and this also bounds
    the number of candidate centers, the total running time is
    $O(k^{3/2}\rho^3)=O(k^2)$. We return the best polygon found by all
    these runs.
\end{proof}

\subsubsection{Smaller turns, faster algorithm}

The same annulus geometry that gives \lemref{turn_angle_2} also gives a
more algorithmic pruning rule.  In a centered thin-annulus run, any
boundary link of the optimum must be almost tangent to the inner disk:
its supporting line cannot cross a disk known to be contained in
$\Popt$.  Thus, from a fixed endpoint, admissible primitive directions
lie in an angular interval of length $O(k^{-1/6})$.  This gives the same
speedup one would expect from the small-turn bound, without adding the
incoming direction to the \DP state.

\begin{lemma}
    \lemlab{dp-annulus_4}%
    The optimal polygon $\Popt$ can be computed in $O(k^{2-1/6})$ time.
\end{lemma}
\begin{proof}
    Use the candidate centers from the proof of \lemref{dp-annulus_3},
    so that one run uses an annulus of width $\rho=O(k^{1/6})$ centered
    within $O(\rho)$ of the true epicenter.  Such a run has
    $\cardin{\VV}=O(\sqrt{k}\rho)=O(k^{2/3})$ good vertices and therefore
    $O(\cardin{\VV}^2)=O(k^{4/3})$ grid-count states, as in
    \lemref{dp_annulus_2}.

    In this run we use the inner-disk tangent filter from the definition
    of~$\G$: a candidate link is discarded if its supporting line meets
    the interior of the inner disk of the annulus.  This cannot remove
    any link of $\Popt$, because that disk is contained in $\Popt$ and
    every boundary link of a convex polygon lies on a supporting line.

    Let the annulus be centered at $q$, with inner radius $r$ and outer
    radius $R$, where $R-r=O(k^{1/6})$ and $R=\Theta(\sqrt{k})$.  Fix a
    point $p\in\VV$, and let $\theta$ be the angle between a candidate
    outgoing vector and the radius vector $p-q$.  If its supporting line
    avoids the interior of the inner disk, then
    \[
        \lenX{p-q}\,|\sin\theta| \geq r .
    \]
    Since $\lenX{p-q}\leq R$, this implies
    $|\sin\theta|\geq r/R=1-O(k^{-1/3})$.  Hence the directed vector lies
    in the union of two angular intervals, each of length
    $O(k^{-1/6})$, around the two tangent directions at~$p$.  The number
    of primitive grid vectors of length at most $\emax=O(k^{1/4})$ in an
    angular interval of length $\alpha=O(k^{-1/6})$ is
    $O(\alpha\emax^2+\emax)=O(k^{1/3})$: cover the corresponding sector of
    radius~$\emax$ by unit grid cells, and use its
    $O(\alpha\emax^2)$ area plus its $O(\emax)$ boundary contribution.

    Thus every grid-count state has only $O(k^{1/3})$ outgoing links in
    the pruned graph.  One centered run takes
    $O(k^{4/3}k^{1/3})=O(k^{5/3})$ time.  Trying the $O(k^{1/6})$
    candidate centers of \lemref{dp-annulus_3} gives total time
    $O(k^{5/3}k^{1/6})=O(k^{2-1/6})$.
\end{proof}

\subsection{Reducing the epicenter location uncertainty}
\seclab{rings_2}

The algorithm of \lemref{dp-annulus_4} still tries
$O(k^{1/6})$ candidate centers.  The true concentric disk sandwich for
$\Popt$ has width only $\WD=O(k^{1/6})$, so the remaining task is to
reduce the one-dimensional uncertainty in the center location.

\begin{lemma}
    \lemlab{generators}%
    For any fixed $\alpha \in (0,1)$, there are constants
    $0<c_1<c_2$ such that at least an $\alpha$-fraction of the edges of
    $\Popt$ have generator length between $c_1 k^{1/6}$ and
    $c_2 k^{1/6}$.
\end{lemma}
\begin{proof}
    By \corref{nv_bound} and \thmref{Andrews},
    $\Popt$ has $\Theta(k^{1/3})$ edges.  The
    oriented edge directions of a convex polygon strictly rotate after
    collinear consecutive pieces are merged into single edges, so there is
    one oriented generator per edge.  By
    \factref{primitive_number}, if $c_1$ is sufficiently small, at most
    $(1-\alpha)/2$ fraction of these generators can be shorter than
    $c_1 k^{1/6}$. Similarly, if more than $(1-\alpha)/2$-fraction of them
    are longer than $c_2 k^{1/6}$, then their total length, which is a
    lower bound on $\perimX{\Popt}$, is larger than (say) $4\sqrt{k}$, which
    is impossible.
\end{proof}

\begin{lemma}
    \lemlab{leftmost-edge}%
    Let $x_\ell$ be the leftmost $x$ coordinate of a vertex of $\Popt$ and
    let $e_\ell=\Popt\cap ( x = x_\ell )$.  Then $e_\ell$ is a vertical
    segment, possibly degenerate, and in fact
    \begin{math}
        \lenX{e_\ell} = \Theta(k^{1/4}).
    \end{math}
\end{lemma}

\begin{proof}
    All edges in this proof are maximal collinear boundary segments, after
    consecutive collinear links have been merged.  The two neighboring
    edges of $e_\ell$ are the two maximal edges sharing an endpoint with
    it.  Let $\lambda=\lenX{e_\ell}$. \corref{edge-improved} implies the upper bound. It remains to prove the matching lower bound $\lambda=\Omega(k^{1/4})$, which also rules out the degenerate case for sufficiently large~$k$.

    Translate $\Popt$ so that
    $e_\ell=\{(0,y):0\le y\le\lambda\}$.  Let
    \[
        Q=\CHX{\bigl(\ZZ^2\cap\Popt\bigr)\setminus e_\ell}.
    \]
    Then $Q$ is a convex grid polygon and, since $Q\subseteq\Popt$,
    it contains exactly the grid points of $\Popt$ except the
    $\lambda+1$ points on~$e_\ell$.

    We compare its perimeter to the geometric cut
    $K=\Popt\cap\{x\ge 1\}$.  Since $Q\subseteq K$, monotonicity of
    perimeter for convex sets gives $\perimX{Q}\leq \perimX{K}$.  If
    the line $x=1$ meets the lower and upper boundary chains at
    $(1,y_-)$ and $(1,y_+)$, write the signed offsets
    \[
        a=y_-,
        \qquad
        b=y_+-\lambda .
    \]
    The cut replaces the left support segment of length~$\lambda$ and the
    two incident boundary pieces of lengths $\sqrt{1+a^2}$ and
    $\sqrt{1+b^2}$ by the vertical segment at $x=1$, whose length is
    $\lambda+b-a$.  Therefore
    \[
        \perimX{\Popt}-\perimX{K}
        =
        \sqrt{1+a^2}+\sqrt{1+b^2}-(b-a)
        =
        \bigl(\sqrt{1+a^2}+a\bigr)
        +
        \bigl(\sqrt{1+b^2}-b\bigr).
    \]
    By \corref{edge-improved}, the two edges adjacent to $e_\ell$ have
    length at most $\emax=O(k^{1/4})$.  Thus
    $\cardin{a},\cardin{b}\leq \emax$.  For any
    $\cardin{t}\leq\emax$,
    \[
        \sqrt{1+t^2}\pm t \geq
        \sqrt{1+t^2}-\cardin{t}
        =
        \frac{1}{\sqrt{1+t^2}+\cardin{t}}
        =
        \Omega(1/\emax).
    \]
    Hence
    \[
        \perimX{\Popt}-\perimX{Q}
        \geq
        \perimX{\Popt}-\perimX{K}
        =
        \Omega(k^{-1/4}).
    \]

    It remains to show that enough cheap insertion edges survive in
    $Q$.  By \corref{nv_bound}, $\nv=\Theta(k^{1/3})$.  Choose the
    constants in the following order.  First choose a sufficiently large
    constant $C_R$ and set $R=C_Rk^{1/6}$.  At most
    $\nv/8$ edges of $\Popt$ have generator length larger than~$R$;
    otherwise their total length would exceed $\lOpt=O(\sqrt{k})$.
    Next choose $\eta>0$ small enough that $\eta R^2\leq \nv/8$ for all
    sufficiently large~$k$.  By \lemref{long}, all but at most
    $\eta R^2$ primitive directions in $\Lambda_R$ are $cR$-long, for
    some constant $c=c(\eta)>0$.  This fixes the positive constant used
    later by \lemref{long_expand}.  Since the oriented edge generators of
    a convex polygon are distinct, at most $\nv/8$ edges of $\Popt$ with
    generator length at most~$R$ fail to be $cR$-long.

    Thus at least $3\nv/4$ edges of $\Popt$ have generator length at
    most~$R$ and are $cR$-long.  Taking every other such edge along the
    boundary gives a pairwise nonadjacent family $\mathcal{E}$ of
    $\Omega(k^{1/3})$ edges.  Let $e$ be a selected edge that is not
    $e_\ell$ and not one of its two neighboring edges, and let $\ell$ be
    its supporting line.  Since edges are maximal exposed faces,
    $\Popt\cap\ell=e$, and all grid points of $\Popt$ lie in one closed
    halfplane bounded by~$\ell$.  Deleting the grid points of $e_\ell$
    removes no point of~$e$, so the endpoints of~$e$ and all lattice points
    on~$e$ remain in the set whose convex hull is~$Q$.  Thus
    $Q\cap\ell=\CHX{(\Popt\cap\ZZ^2\setminus e_\ell)\cap\ell}=e$, and
    $e$ is still an actual exposed edge of~$Q$, not merely a supporting
    subsegment.  We keep only these surviving selected edges.  They remain
    pairwise nonadjacent as edges of~$Q$: away from the cut their incident
    boundary arcs are unchanged, and near the cut the replacement chain is
    incident only to the two removed neighboring edges, which were not kept.

    We discarded at most three selected edges, so the surviving subfamily
    still has size $\Omega(k^{1/3})$.  Since $\lambda=O(k^{1/4})$, at least
    $\lambda+1$ surviving edges remain for all sufficiently large~$k$.  By
    \lemref{long_expand}, we can insert $\lambda+1$ grid points on these
    surviving edges at total perimeter cost
    $O((\lambda+1)/R^3)=O((\lambda+1)k^{-1/2})$.  These inserted points
    are new points outside~$Q$.  They are distinct by the projection
    argument in \lemref{long_expand}, and none of them is one of the
    removed points of~$e_\ell$: each insertion is on the exterior side of
    a supporting line of a surviving non-neighbor edge, while every point
    of~$e_\ell$ lies in $\Popt$ on the non-exterior side of that line.
    The resulting polygon
    contains at least $k$ grid points.  By \obsref{trim-to-k}, it
    contains an exactly feasible $k$-point polygon with no larger
    perimeter.  By optimality of $\Popt$,
    \[
        O((\lambda+1)k^{-1/2})
        \geq
        \Omega(k^{-1/4}),
    \]
    and hence $\lambda=\Omega(k^{1/4})$.  All constants above are
    effective once the constants in \lemref{long} and
    \lemref{long_expand} are fixed; choosing $c_\ell>0$ smaller than the
    resulting hidden constant gives the certified lower bound
    $\lambda\geq c_\ell k^{1/4}$ used later.
\end{proof}

\begin{figure}[t]
    \centering%
    {\includegraphics{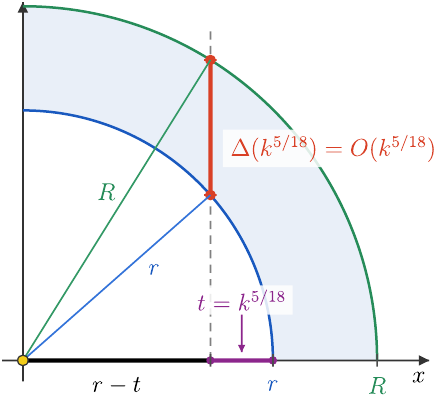}}
    \caption{Illustration of \lemref{one_ring}.}
    \figlab{one_ring}
\end{figure}
\begin{lemma}
    \lemlab{one_ring}%
    Let $r < R$ be two radii with $r,R=\sqrt{k/\pi}+O(k^{1/6})$ and
    $R-r \leq \WD = O(k^{1/6})$. For $t \geq 0$, let
    $\Sin(t)= (x=r-t) \cap \diskY{\origin}{r}$ and
    $\Sout(t)= (x=r-t) \cap \diskY{\origin}{R}$ be two vertical segments,
    where $\diskY{\origin}{r}$ denotes the disk of radius $r$ centered at
    the origin.  Let $\Delta(t)$ be the vertical distance between the top
    endpoints of $\Sin(t)$ and $\Sout(t)$.  For any
    $t \in [1/2, 3/2]k^{5/18}$, we have that
    $\lenX{\Sout(t)} = \Theta(k^{7/18})$ and $\Delta(t) = O( k^{5/18})$.
\end{lemma}
\begin{proof}
    See \figref{one_ring}. We have that
    \begin{align*}
        \lenX{\Sout(t)}
        &=%
        2 \sqrt{R^2 - (r-t)^2}
        =
        2 \sqrt{ (R+r -t) ( R -r +t)}
        \\&
        =
        \Theta\Bigl( \sqrt{ k^{1/2}  (k^{1/6} + k^{5/18})}\Bigr)
        =%
        \Theta( k^{1/4+5/36})
        =%
        \Theta( k^{7/18}).
    \end{align*}
    Similarly, we have
    \begin{align*}
        \Delta(t)
        &=%
        \sqrt{R^2 - (r-t)^2} -
        \sqrt{r^2 - (r-t)^2}
        \\&
        =
        \frac{R^2 - r^2}{ \sqrt{R^2 - (r-t)^2} +
           \sqrt{r^2 - (r-t)^2}}
        \\&
        \leq
        \frac{2 R \WD }{2\sqrt{t(2r-t)}}
        \leq
        \frac{ R \WD }{\sqrt{t r}}
        \leq
        \frac{ 2\sqrt{R} \WD }{\sqrt{t}}
        =%
        O\pth{\frac{k^{1/4 + 1/6}}{ k^{5/36}} }
        \\&
        =
        O\pth{k^{10/24 - 5/36}}
        =
        O\pth{k^{(15-5)/36}}
        =
        O\pth{k^{5/18}}.
    \end{align*}
\end{proof}

\begin{lemma}
    \lemlab{leftmost_midpoint}
    Translate $\Popt$, so its \defrefY{epicenter}{epicenter} is at the
    origin.  Let $x_r$ be the largest $x$-coordinate of a point of
    $\Popt$, and set $s_r=\Popt\cap (x=x_r)$.  Then:
    \begin{compactenumI}
        \item $\sqrt{k/\pi} - \Gamma \leq x_r \leq \sqrt{k/\pi} + \Gamma$, where $\Gamma=C_\Gamma k^{1/6}$,
        \item $s_r$ is of length $\Theta(k^{1/4})$,
        \item the midpoint of $s_r$ is contained in the segment $x_r \times [-\nabla, \nabla]$, where $\nabla = O(k^{5/18})$.
    \end{compactenumI}
\end{lemma}

\begin{proof}
    (I) By \thmref{improved-gap}, $\partial \Popt$ lies in an annulus of
    width $\WD = O(k^{1/6})$, between disks of radii $\rin$ and $\Rout$
    centered at the origin.  Pick's theorem and $b=O(k^{1/3})$ imply
    $\areaOpt=k-O(k^{1/3})$, and hence
    $\rin,\Rout=\sqrt{k/\pi}+O(k^{1/6})$.  Since the disk
    $\diskY{\origin}{\rin}$ is contained in $\Popt$ and
    $\Popt\subseteq\diskY{\origin}{\Rout}$, the right supporting
    coordinate satisfies $\rin\leq x_r\leq \Rout$.  This proves (I).

    \medskip%
    (II) follows from \lemref{leftmost-edge} after reflecting the polygon
    through the $y$-axis.  The lemma is unchanged by translating the
    underlying lattice and by this reflection.

    \medskip%
    (III) After the translation, the grid is a translated lattice
    $L=\ZZ^2-c$ and is still invariant under integer vertical shifts.
    The polygon $\Popt$ is optimal for this translated lattice, since the
    original problem is translation invariant.  Let $\XSet$ be the set of
    $x$-coordinates of~$L$.  Let $x^- \in \XSet$ be the maximum
    number smaller than $\rin-k^{5/18}$, and define
    $s(\xi)=\Popt\cap\{x=x^-+\xi\}$.  For
    $t=\rin-x^-$ we have
    $k^{5/18}\leq t\leq k^{5/18}+1$.  By \lemref{one_ring}, the segment
    $s(0)$ is contained in the outer chord
    $\diskY{\origin}{\Rout}\cap\{x=x^-\}$, which is symmetric about the
    $x$-axis, and its two endpoints are within vertical distance
    $O(t)$ from the corresponding endpoints of that chord.  Hence the
    midpoint of $s(0)$ has $y$-coordinate $m(0)=O(t)$.

    Let $s_i=s(i)$, for $i=0,\ldots,W=x_r-x^-$.  Since
    $x_r\leq\Rout=\rin+O(k^{1/6})$, we have
    $W=O(k^{5/18})$.  Let $m(i)$ be the $y$-coordinate of the midpoint
    of $s_i$.  We claim that
    \[
        \cardin{m(i)-m(i-1)}\leq \frac12
        \qquad\text{for every }i=1,\ldots,W.
    \]

    Fix $i$, and write the vertical chords as
    $s_{i-1}=[a^-,b^-]$ and $s_i=[a^+,b^+]$, with $a^\pm$ on the lower
    boundary chain and $b^\pm$ on the upper boundary chain.  Since
    $x^-+i-1$ and $x^-+i$ are consecutive columns of the translated
    lattice, no vertex of the grid polygon $\Popt$ lies in the open slab
    between them.  Hence the two boundary pieces of $\Popt$ crossing this
    slab are exactly the straight bridges $a^-a^+$ and $b^-b^+$.

    For an integer $h$, form a closed polygonal region $P_h$ as follows.
    Let $P^-=\Popt\cap\{x\leq x^-+i-1\}$ and
    $P^+=\Popt\cap\{x\geq x^-+i\}$.  Keep $P^-$ fixed, translate $P^+$
    by the lattice vector $(0,h)$, and replace the two slab bridges by the
    segments from $a^-$ to $a^++(0,h)$ and from $b^-$ to
    $b^++(0,h)$.  The lattice-point count with respect to~$L$ is
    unchanged.  Use the half-open ownership rule that lattice points in the
    left boundary column belong to $P^-$ and lattice points in the right
    boundary column belong to $P^+$.  The map $z\mapsto z+(0,h)$ is a
    bijection of~$L$, the open slab between the two consecutive lattice
    columns contains no point of~$L$, and the new bridge interiors also lie
    in this open slab.  Hence $P_h\cap L$ has exactly as many lattice
    points as $\Popt\cap L$.

    The boundary length of $P_h$ differs from $\perimX{\Popt}$ only by
    replacing the two original bridge lengths by the two new bridge
    lengths.  The region $P_h$ need not be convex.  However
    $K_h=\CHX{P_h\cap L}$ contains all lattice points of $P_h$ and hence at
    least $k$ points of~$L$.  Also, the perimeter of the convex hull of a
    closed curve is at most the length of the curve (for example by
    Cauchy's projection formula), so
    $\perimX{K_h}\leq\perimX{\CHX{P_h}}$, and this is at most the
    boundary length of~$P_h$.
    Applying \obsref{trim-to-k} to the translated lattice~$L$ and using the
    optimality of $\Popt$ for~$L$, no integer shift can reduce the total
    bridge length between these two columns.

    Let $\delta=m(i)-m(i-1)$.  If the lengths of $s_{i-1}$ and $s_i$ are
    $\ell_{i-1}$ and $\ell_i$, the sum of the two bridge lengths after
    shifting by $h$ is
    \[
        G(\delta+h)
        =
        \sqrt{1+(\delta+h+a)^2}
        +
        \sqrt{1+(\delta+h-a)^2},
        \qquad
        a=\frac{\ell_i-\ell_{i-1}}{2}.
    \]
    The function $G$ is even and strictly increasing as a function of
    $\cardin{\delta+h}$: for $x>0$ its derivative is
    \[
        \frac{x+a}{\sqrt{1+(x+a)^2}}
        +
        \frac{x-a}{\sqrt{1+(x-a)^2}}>0,
    \]
    since $y/\sqrt{1+y^2}$ is odd and strictly increasing.  If
    $\cardin{\delta}>1/2$, an integer $h$ can be chosen with
    $\cardin{\delta+h}<\cardin{\delta}$, contradicting optimality.
    Therefore $\cardin{m(i)-m(i-1)}\leq 1/2$.

    Consequently,
    \[
        \cardin{m(W)}
        \leq
        \cardin{m(0)}+\sum_{i=1}^W \cardin{m(i)-m(i-1)}
        \leq
        O(t)+W/2
        =
        O(k^{5/18}).
    \]
\end{proof}

\subsubsection{The improved algorithm}

Once a bottom support grid point of $\Popt$ is fixed at the origin,
\lemref{in} places the epicenter in a region of height $O(k^{1/6})$.
Equivalently, before this final anchoring, we may rotate the picture so
that the relevant support line is vertical and apply
\lemref{leftmost_midpoint}; this reduces the remaining horizontal
uncertainty to $O(k^{5/18})$.  The chosen anchor need not be the midpoint
of that support edge, but the edge itself has length $O(k^{1/4})$ by
\corref{edge-improved}, which is absorbed by the $O(k^{5/18})$ error
term.  Thus, one can compute a set of
$O(k^{5/18}/ k^{1/6}) = O(k^{1/9})$ candidate centers, spaced so that one
of them is at distance $O(k^{1/6})$ from the epicenter.  Modifying the
algorithm of \lemref{dp-annulus_4} to use $O(k^{1/9})$ candidate centers
instead of $O(k^{1/6})$ gives:

\begin{lemma}
    \lemlab{dp-annulus_5}%
    The optimal polygon $\Popt$ can be computed in $O(k^{1+7/9})$ time.
\end{lemma}
\begin{proof}
    For any fixed bottom-support anchor, the preceding paragraph gives an
    explicitly constructible rectangle that contains the epicenter, has
    height $O(k^{1/6})$, and has width $O(k^{5/18})$.  Put a grid of
    candidate centers in this rectangle with spacing $\Theta(k^{1/6})$.
    The number of candidates is $O(k^{5/18}/k^{1/6})=O(k^{1/9})$, and one
    of them is within $O(k^{1/6})$ of the true epicenter.

    For one such candidate center, use the same annulus and tangent
    pruning as in the proof of \lemref{dp-annulus_4}.  The annulus has
    $O(k^{2/3})$ good vertices, hence $O(k^{4/3})$ grid-count states, and
    each state has $O(k^{1/3})$ outgoing links after the tangent filter.
    Thus one centered run costs $O(k^{5/3})$.  Multiplying by the
    $O(k^{1/9})$ candidate centers gives
    \[
        O(k^{5/3+1/9}) = O(k^{1+7/9}).
    \]
\end{proof}

\subsection{Speeding up transitions via monotone queues}
\seclab{monotonicity}

We now give a final, more technical speedup.  The geometric reductions
above already give \lemref{dp-annulus_5}; here we improve the transition
processing by batching transitions that share the same horizontal
displacement and exploiting a monotonicity property of the Euclidean
distance.  Throughout this subsection, $\emax=\Theta(k^{1/4})$ denotes the
algorithmic edge-length cap guaranteed by \corref{edge-improved}, with the
constant chosen large enough not to remove an edge of~$\Popt$.

\subsubsection{Explicit grid-count update for primitive edges}

A \emphi{column} is the set of points in $\VV$ sharing a common
$x$-coordinate.  A \emphi{column pair} is an ordered pair of columns
$(x_0,x_1)$ that can occur as source and target columns in the polar
topological order, with $0<\cardin{x_1-x_0}\leq \emax$.  The case
$x_0=x_1$ corresponds to vertical primitive vectors $(0,\pm1)$ and is
processed directly within the same time bound.

The nonprimitive source and sink edges incident to the origin copies are
processed separately by their gcd grid counts.  There are only
$O(\cardin{\VV})$ such graph edges in a centered run.  The source
relaxations take $O(\cardin{\VV})$ time, and the sink relaxations take
$O(\cardin{\VV}^2)$ time over all grid-count states, which is absorbed by
the bounds below.  The batched transitions in this subsection are the
internal primitive edges of~$\G$.

Consider a transition from source $s = (x_0, i)$ with grid count $\nG_s$ to target $t = (x_1, j)$ with grid count $\nG_t$.  The \DP adds the triangle $\triangle = \triangle \origin s t$.  The new grid points are those in~$\triangle$ not already counted (i.e., excluding the origin, the vertex~$s$, and the interior grid points of the shared diagonal $\origin s$).

\begin{lemma} %
    \lemlab{grid-update}%
    For a \emph{primitive} edge $s \to t$, i.e., $\rho = \gcdY{(x_1 - x_0)}{(j - i)} = 1$, the grid-count update is
    \begin{equation*}
        \nG_t - \nG_s
        =
        \half (x_0 j - x_1 i)
        - \half (\gcdY{x_0}{i})
        + \half (\gcdY{x_1}{j})
        + \half.
    \end{equation*}
\end{lemma}

\begin{proof}
    The triangle $\triangle = \triangle \origin s t$ has signed area $\areaC_\triangle = \frac{1}{2}(x_0 j - x_1 i)$.  And $\triangle$ has
    \begin{align*}
        b_\triangle
        &=
        (\gcdY{x_0}{i})
        + \rho
        + (\gcdY{x_1}{j})
        =
        (\gcdY{x_0}{i}) + 1 + (\gcdY{x_1}{j})
    \end{align*}
    grid points on its boundary.  By Pick's theorem, the total number of grid points in or on~$\triangle$ is $\xi = \areaC_\triangle + \half b_\triangle + 1$.  Subtracting the already-counted points---the origin, the vertex~$s$, and the $(\gcdY{x_0}{i}) - 1$ interior points of $\origin s$, subtracting $(\gcdY{x_0}{i}) + 1$ in total, implies
    \begin{align*}
        \nG_t - \nG_s
        &=%
        \xi - (\gcdY{x_0}{i}) - 1
        \\&
        =
        \areaC_\triangle + \half b_\triangle + 1 - \pgcdY{x_0}{i} - 1\\
        &=%
        \half (x_0 j - x_1 i) + \half \pgcdY{x_0}{i} + \half
        + \half \pgcdY{x_1}{j} - \pgcdY{x_0}{i}%
        \\&%
        =%
        \half (x_0 j - x_1 i)- \half \pgcdY{x_0}{i}
        + \half \pgcdY{x_1}{j} + \half .
    \end{align*}
\end{proof}

\subsubsection{The conserved index}

\begin{defn}
    \deflab{potential}%
    Given a transition in the \DP from $(s,\nG_s) \to (t,\nG_t)$, its (source) \emph{potential} and target \emph{potential} are, respectively,
    \begin{equation*}
        \varphi(s,\nG_s, x_1) = \nG_s - \half  x_1 i
        - \half \pgcdY{x_0}{i}%
        \quad\text{and}\quad
        \overline{\varphi}(t,\nG_t, x_0)
        =
        \nG_t - \half x_0 j
        - \half \gcdY{x_1}{j}
        - \half. %
    \end{equation*}
    Note that the two potentials are the same by \lemref{grid-update}. In particular, the \emphi{potential} of the transition is
    \begin{math}
        \sigma = \varphi(s,\nG_s, x_1) = \overline{\varphi}(t,\nG_t, x_0).
    \end{math}
\end{defn}

Thus, the transition potential depends only on the source state and the $x$-coordinate of the target state (or on the target state and the $x$-coordinate of the source). The above readily implies the following:

\begin{lemma}
    \lemlab{conserved}%
    For a fixed value of the potential $\sigma$, we have:
    \begin{compactenumI}
        \item %
        $\nG_s = \sigma + \frac{1}{2}\, x_1 i + \frac{1}{2}\pgcdY{x_0}{i}$.

        \item \begin{math} \nG_t%
            =%
            \sigma + \frac{1}{2} x_0 j + \frac{1}{2}\pgcdY{x_1}{j} + \frac{1}{2}.
        \end{math}

        \item For two source states $(s, \nG_s), (s', \nG_{s'})$ sharing
        the same $x$-axis, with $s=(x_0, i)$ and $s'=(x_0, i')$, and
        consider their transition into the same target state
        $(t, \nG_t)$, where $t= (x_1, j)$. For,
        $\Delta = \varphi(s, \nG_s,x_1) - \varphi(s', \nG_{s'},x_1) $, we
        have
        \begin{equation*}
            \Delta =
            (\nG_{s} - \nG_{s'}) - \half x_1(i - i') -
            \half \bigl(\pgcdY{x_0}{i} - \pgcdY{x_0}{i'}\bigr),
        \end{equation*}
        which is independent of the target height $j$.
    \end{compactenumI}
\end{lemma}

\subsubsection{Monotonicity of the cost curves}

\newcommand{\perC}{\overline{\tau}}

The dynamic programming maintains for every state $(s,\nG_s)$ the best perimeter of the polygon found, where $s$ is the last vertex before going back to the origin. Let $\perC( s, \nG_s)$ denote this value.

For a transition from $s=(x_0, i)$ (with current perimeter $\perC$) to $t=(x_1, j)$, the new perimeter is $\perC' = \perC - \dY{\origin}{s} + \dY{\origin}{t} + \dY{s}{t}$.  Since $\dY{\origin}{t}$ depends only on the target, minimizing $\tau'$ over sources at column $x_0$ (for a fixed potential $\sigma$) is equivalent to minimizing
\begin{align}
    F_i(j)
    &=
    C_i + d_i(j),
\end{align}
where
\begin{math}
    C_i = \perC\bigl(s, \nG_s) - \dY{\origin}{(x_0, i)} %
    = \perC\bigl(s, \sigma + \half x_1 i + \half \pgcdY{x_0}{i} \bigr) - \dY{\origin}{(x_0, i)},
\end{math}
and
\begin{math}
    d_i(j) = \dY{s}{t} = \sqrt{(x_1 - x_0)^2 + (j - i)^2\,}.\Bigr.
\end{math}
Here $C_i$ does not depend on $t$.

\begin{lemma} %
    \lemlab{crossing}%
    Fix two distinct columns and put $a=\cardin{x_1-x_0}>0$.  For source
    heights $i_0 < i_1$, the difference
    $F_{i_0}(j)-F_{i_1}(j)$ is strictly increasing in~$j$.  In particular,
    the two curves cross at most once.
\end{lemma}

\begin{proof}
    The constants $C_{i_0}$ and $C_{i_1}$ do not depend on~$j$, so it
    suffices to prove that $d_{i_0}(j)-d_{i_1}(j)$ is strictly
    increasing.  Its derivative is
    \[
        \frac{j-i_0}{\sqrt{a^2+(j-i_0)^2}}
        -
        \frac{j-i_1}{\sqrt{a^2+(j-i_1)^2}}.
    \]
    The function $x/\sqrt{a^2+x^2}$ is strictly increasing for
    $a>0$, and $j-i_0>j-i_1$.  Hence this derivative is positive
    for every~$j$, proving strict monotonicity.  In particular, the
    difference between the two curves has at most one zero.
\end{proof}

Each column pair $(x_0, x_1)$ and a fixed potential $\sigma$, form a \emphi{batch}. It is the set
\begin{equation*}
    \Set{ \bigl((x_0, i), \nG_s \bigr)}{\nG_s = \sigma + \half x_1 i + \half\pgcdY{x_0}{i}}.
\end{equation*}

\begin{corollary} %
    \corlab{monotone-opt}%
    Within a batch, let $i^*(j)$ denote the source height minimizing $F_i(j)$ among all sources in the batch.  Then $i^*(j)$ is non-decreasing in~$j$.
\end{corollary}

\begin{proof}
    Suppose $i^*(j_0) = i_a$ and $j_1 > j_0$.  For any $i_b < i_a$: since $F_{i_b}(j) - F_{i_a}(j)$ is strictly increasing (\lemref{crossing}) and $F_{i_b}(j_0) \geq F_{i_a}(j_0)$ (optimality of $i_a$ at $j_0$), we have $F_{i_b}(j_1) > F_{i_a}(j_1)$.  So $i_b$ cannot be optimal at~$j_1$, giving $i^*(j_1) \geq i_a = i^*(j_0)$.
\end{proof}

We now use the following standard corollary.

\begin{corollary}
    \corlab{envelope}%
    Given a set of $q$ source curves $\{F_i\}_{i=1}^{q}$ inserted in
    increasing source-height order, and $r$ target-height queries made in
    increasing target-height order within one batch, all queried target
    minima can be computed online in $O(q+r)$ time.  The same statement
    holds after replacing all heights by their negatives.
\end{corollary}
\begin{proof}
    By \corref{monotone-opt}, the optimal source index is non-decreasing
    as target heights are queried in order.  Maintain a deque of candidate
    sources.  When inserting source $i'$ at the back, remove any source
    $i_b$ at the back for which the crossing $F_{i_b}=F_{i'}$ occurs no
    later than the crossing $F_{i_b}=F_{i_a}$, where $i_a$ is immediately
    ahead of $i_b$ in the deque; after that point $i_b$ can never be
    optimal.  When querying target~$j$, advance the front past any source
    that has been overtaken.  The insertion and query orders are monotone,
    so each source enters and leaves the deque at most once, and each
    queried target is queried once.  The total time is $O(q+r)$.
\end{proof}

\paragraph*{A note on column order.}
For a nonvertical column, the polar sweep orders heights increasingly when
$x>0$ and decreasingly when $x<0$; on the column $x=0$, ties are ordered
by distance from the origin, hence by increasing height in our upper
halfplane.  Same-side column pairs are handled online in these
polar-induced orders.  If both columns are negative, this is just the
change of variables $\tilde i=-i$, $\tilde j=-j$, so
\lemref{crossing} and \corref{monotone-opt} apply verbatim.

The remaining admissible nonvertical pairs cross the $y$-axis in the
forward polar direction: positive-to-zero, positive-to-negative, or
zero-to-negative.  For such a pair, every source state in the batch is
finalized before any target state in the batch is processed.  We therefore
build the lower envelope statically before the first target query.  If the
target column is $x_1=0$, its target queries are in increasing height, and
we use the original heights.  If $x_1<0$, its target queries are in
decreasing height, and we use the reflected coordinates
$(\tilde i,\tilde j)=(-i,-j)$; the sources are sorted by increasing
$\tilde i$ when the static envelope is built.  The reverse mixed
directions, namely zero-to-positive, negative-to-zero, and
negative-to-positive, have no admissible internal edge because they violate
the polar order.  The case $x_0=x_1$ consists of vertical primitive
vectors $(0,\pm1)$ and is processed directly.

\begin{lemma}
    \lemlab{eligible-envelope}%
    Fix a column pair $(x_0,x_1)$, put $a=\cardin{x_1-x_0}>0$, and fix a potential value $\sigma$.  Suppose $q$ source states are inserted into this batch and $r$ target states are actually queried.  The transitions of this batch that are actual edges of the visibility graph~$\G$ can be processed in
    \begin{math}
        (q+r)\,\emax^{o(1)}
    \end{math}
    time.
\end{lemma}

\begin{proof}
    Choose the height coordinate used for this batch as described in the preceeding
    column-order paragraph: either the original height, or the reflected
    height $\tilde y=-y$.  The potential value is still computed in the
    original coordinates; the chosen coordinate is only for sorting,
    residue filtering, and the lower-envelope deque.  In the rest of this
    proof, $i$ and $j$ denote these chosen height coordinates.  In the
    reflected case the Euclidean distance and the primitiveness condition
    are unchanged because $j-i=-(\tilde j-\tilde i)$ in the original
    coordinates.  Let $J$ be the sorted list of all target heights in the
    column~$x_1$ in this chosen coordinate; write
    $M=\cardin{J}$.  The value $M$ enters only through logarithmic
    decomposition factors, while $r$ below denotes the number of target
    states actually queried in this batch.

    Fix a source state with source point $s=(x_0,i)$, using the chosen
    height coordinate for the notation.  We first ignore primitiveness.
    The admissible target heights form a constant number of intervals in
    the sorted list~$J$.  The length constraint is
    $\cardin{j-i}\leq \sqrt{\emax^2-a^2}$, hence one interval.  The
    polar-order constraint $s\prec t$ is also an interval on a fixed target
    column, with the same distance tie-breaking used in the definition of
    $\G$.  For visibility, as $t=(x_1,j)$ moves up a fixed vertical column,
    the directed line $st$ rotates monotonically.  Since
    $\mathcal{B}$ is convex and $s\notin\mathcal{B}$ for internal edges,
    the lines through $s$ that touch $\mathcal{B}$ are the two tangents
    (possibly coincident in a degenerate case); between two consecutive
    tangent directions the truth value of ``$\mathcal{B}$ lies on the swept
    side of $st$ and the supporting line does not cut through
    $\mathcal{B}$ except possibly at~$\origin$'' is constant.  Touching a
    vertex or lying along a supporting edge only changes whether an endpoint
    of an interval is included, which is immaterial for the canonical
    decomposition.  The same monotone-rotation argument applies to the
    inner-disk tangent filter: the forbidden directions are the open angular
    interval between the two tangents from~$s$ to $\Din$ (or empty if no
    tangent obstruction exists).  Thus visibility and tangency introduce
    only $O(1)$ additional breakpoints.  Intersecting these constant-many
    constraints gives $O(1)$ intervals of target heights for this source.

    The segment tree over~$J$ is implicit: nodes are created only when a
    copied source or an actual target query touches them.  A source
    interval is decomposed into $O(\log M)$ canonical nodes, so each source
    copy is stored only where every target in the node satisfies all
    nonprimitive eligibility constraints for that source.  A target query
    visits the $O(\log M)$ ancestors of its leaf, and therefore sees exactly
    the sources whose nonprimitive constraints allow that target height.

    It remains only to enforce $\gcdY{a}{(j-i)}=1$ in these chosen
    coordinates.  Let $p_1,\ldots,p_t$ be the distinct prime divisors
    of~$a$.  In each target-tree node, use the product of segment trees on
    the residue sets $\ZZ/p_\ell\ZZ$: store a source according to the
    residue vector $(i\bmod p_1,\ldots,i\bmod p_t)$, and route a target~$j$
    to the canonical decomposition of
    $\prod_\ell((\ZZ/p_\ell\ZZ)\setminus\{j\bmod p_\ell\})$.  This creates
    $O(\prod_\ell \log p_\ell)$ source and query copies.  Since
    $\prod_\ell p_\ell\leq a\leq\emax$, for every fixed $\varepsilon>0$ we
    have
    $\prod_\ell O(\log p_\ell)=O_\varepsilon(\emax^\varepsilon)=\emax^{o(1)}$.

    Each resulting canonical subproblem preserves the required order.  For
    same-side online batches, a segment-tree node receives a subsequence of
    the globally ordered source insertions, and the residue filter only
    splits this subsequence into further subsequences; target queries reach
    ancestor nodes, and the residue decomposition again keeps subsequences
    of the globally ordered target queries.  Hence every online canonical
    subproblem satisfies \corref{envelope}.  For a forward mixed batch, all
    sources are already finalized before the first target query.  For each
    touched canonical subproblem we sort its copied sources in the chosen
    height coordinate, build the deque/static lower envelope in linear time
    after sorting, and then answer the target queries in the chosen target
    order.  The total sorting cost is bounded by the number of copied
    sources times an additional $O(\log k)$ factor, since no copied source
    list has size larger than the total number of states.

    Across all canonical copies, the total number of copied source
    insertions, including allocation and static-envelope construction, is
    bounded by
    \[
        O(q\log M\log k)\prod_{\ell=1}^tO(\log p_\ell),
    \]
    where the final $O(\log k)$ factor is only needed for static mixed
    batches.  The total number of target-query copies is
    \[
        O(r\log M)\prod_{\ell=1}^tO(\log p_\ell).
    \]
    Summing gives
    \[
        O\bigl((q+r)\log M\log k\bigr)
        \prod_{\ell=1}^t O(\log p_\ell)
        =
        (q+r)\emax^{o(1)} ,
    \]
    since the algorithmic cap in this subsection is
    $\emax=\Theta(k^{1/4})$, so $M\leq O(k)$ implies
    $\log M, \log k=\emax^{o(1)}$.  The canonical decompositions are exact,
    so no ineligible edge is used, and no eligible edge is missed.
\end{proof}

\subsubsection{The improved algorithm}

The transition reorganization applies to any annulus run constructed as
in \lemref{dp-annulus}.  The only parameters that enter the running-time
bound are the number of good vertices and the maximum allowed link
length.

\begin{lemma}
    \lemlab{queue-annulus}%
    Let $\VV$ be the good vertex set of one annulus run, and suppose all
    links used by the run have length at most~$\emax$.  The dynamic
    program over~$\VV$ can be evaluated in
    \[
        O\bigl(\cardin{\VV}^2\,\emax^{1+o(1)}\bigr)
    \]
    time.
\end{lemma}

\begin{proof}
We use the same visibility graph $\G = (\VV, \EE)$ and \DP{} state space
$(s, \nG)$ as in \lemref{dp-annulus}, but reorganize the transitions by
column pair.  States are processed in the same polar topological order as
before.  For same-side online batches, when a state becomes final, it is
inserted into the batch structures for all future target columns within
horizontal distance~$\emax$; when a target state is processed, it queries
only these already-finalized sources.  For forward mixed batches crossing
the $y$-axis, finalized sources are first collected in lists; immediately
before the first target query for such a batch, all possible source states
for that batch are already final, so the canonical envelope structures are
built statically in the order from the column-order paragraph.

\paragraph*{Batching.} %
For each ordered column pair $(x_0, x_1)$ and each value of the
potential $\sigma$ (\defref{potential}), we maintain a \emphi{batch}.
When a source state $\bigl((x_0, i),\nG_s\bigr)$ is finalized, its value
determines, for this target column~$x_1$, the unique potential
\[
    \sigma=\nG_s-\frac12 x_1 i-\frac12\pgcdY{x_0}{i}.
\]
The source state is inserted into that online batch, or appended to the
static source list if the batch is a forward mixed batch.  When a target state
$\bigl((x_1,j),\nG_t\bigr)$ is processed, the source column~$x_0$
determines the unique potential
\[
    \sigma=\nG_t-\frac12 x_0j-\frac12\pgcdY{x_1}{j}-\frac12,
\]
and the target queries only that batch.  In a static mixed batch, the
first such query triggers the one-time build from the collected source
list.  Thus batches are paid for by actual source insertions and actual
target queries.  By
\lemref{eligible-envelope}, a batch with $q$ inserted sources and $r$
queried targets costs $(q+r)\emax^{o(1)}$ while considering exactly the
transitions present in~$\G$.

\paragraph*{Total cost.}  Rather than counting potential values, we count
actual insertions and queries.  Each finalized source state is inserted
into the batch structures for $O(\emax)$ possible target columns, and each
insert costs $\emax^{o(1)}$ amortized after the canonical decompositions
of \lemref{eligible-envelope}.  Since there are $O(\cardin{\VV}^2)$
states, the total insertion work is
$O(\cardin{\VV}^2\emax^{1+o(1)})$.

Similarly, for each target state $\bigl((x_1, j), \nG_t\bigr)$ and each
source column~$x_0$ with $|x_1 - x_0| \leq \emax$, the conserved index
determines a unique batch to query.  The total query work is bounded by
the number of (target state, source column) pairs, times the same
$\emax^{o(1)}$ overhead:
\begin{equation*}
    \sum_{{(x_1, j, \nG_t)}}
    \cardin{\Set{x_0 }{ |x_1 - x_0| \leq \emax}}
    \leq
    O(\cardin{\VV}^2) \cdot O(\emax)\cdot \emax^{o(1)}
    =
    O(\cardin{\VV}^2 \cdot \emax^{1+o(1)}),
\end{equation*}
since there are $O(\cardin{\VV}^2)$ target states (\lemref{dp-annulus}) and each has $O(\emax)$ candidate source columns.  Insertions and queries together give
$O(\cardin{\VV}^2 \cdot \emax^{1+o(1)})$.

\paragraph*{Correctness.}
The eligible-envelope reorganization changes how predecessor minima are
found, not the set of \DP transitions.  The topological order guarantees
that a target state queries only finalized predecessor states: an
uninserted future state is after the target in the polar order and
therefore cannot be a predecessor, while forward mixed batches are built
only after all of their sources have become final.  Every
query is restricted to actual edges of~$\G$, including the polar-order,
primitive, visibility, length, and inner-disk tangent constraints.  Thus
the recurrence being evaluated is exactly the original acyclic \DP
recurrence, and the optimum found is unchanged.
\end{proof}

Applying this lemma to the annulus from \lemref{in} gives the following
standalone speedup.

\begin{theorem}
    \thmlab{dp-queue}%
    The minimum-perimeter convex grid polygon enclosing $k$ grid points can be computed in $O(k^{2 - 1/12 + o(1)})$ time.
\end{theorem}

\begin{proof}
    By \lemref{visibility}, the annulus construction has
    $\cardin{\VV}=O(k^{5/6})$.  By \corref{edge-improved}, we may take
    $\emax=O(k^{1/4})$.  Applying \lemref{queue-annulus} gives
    \[
        O\bigl(\cardin{\VV}^2\,\emax^{1+o(1)}\bigr)
        =
        O\bigl(k^{5/3} k^{1/4+o(1)}\bigr)
        =
        O(k^{23/12+o(1)})
        =
        O(k^{2-1/12+o(1)}).
    \]
\end{proof}

\begin{lemma}%
    \lemlab{centered-dp}
    Assume the optimal polygon $\Popt$ is anchored so that the origin is a
    boundary grid point on a bottom supporting line and
    $\Popt\subseteq\{y\geq0\}$.  Let $c$ be the
    \defrefY{epicenter}{epicenter} of $\Popt$.  Let
    $\Gamma=C_\Gamma k^{1/6}$ be an explicit upper bound on the
    minimum-width concentric disk sandwich for $\Popt$, and let
    $r_0=\sqrt{k/\pi}$.  For any candidate point~$q$, the annulus \DP
    restricted to the annulus centered at~$q$ with radii $r_0-C\Gamma$ and
    $r_0+C\Gamma$ can be evaluated in $O(k^{1 + 7/12+o(1)})$ time using
    \lemref{queue-annulus}.  If, in addition, $\dY{q}{c}=O(\Gamma)$, then
    this annulus contains every boundary grid point of $\Popt$, and the
    annulus \DP contains the boundary path of~$\Popt$.
\end{lemma}

\begin{proof}
    The annulus has radius $\Theta(\sqrt{k})$ and width $O(\Gamma)$, so the area estimate from \lemref{visibility} gives $\cardin{\VV}=O(\sqrt{k}\,\Gamma)=O(k^{2/3})$.  By \corref{edge-improved}, the run may use $\emax=O(k^{1/4})$.  Applying \lemref{queue-annulus} gives
    \[
        O\bigl(\cardin{\VV}^2\,\emax^{1+o(1)}\bigr)
        =
        O\bigl(k^{4/3} k^{1/4+o(1)}\bigr)
        =
        O(k^{19/12+o(1)}),
    \]
    which is the claimed running time.

    For containment, \thmref{improved-gap} and Pick's theorem imply $\diskY{c}{\rin}\subseteq \Popt\subseteq\diskY{c}{\Rout}$ with $\Rout-\rin=O(k^{1/6})=O(\Gamma)$ and $\rin,\Rout=r_0+O(\Gamma)$, exactly as in the proof of \lemref{in}.  Therefore, if $\lenX{q-c}=O(\Gamma)$, then for a sufficiently large constant $C$,
    \[
        \diskY{q}{r_0-C\Gamma} \subset \operatorname{int}\bigl(\diskY{c}{\rin}\bigr) \subseteq \Popt \subseteq \diskY{c}{\Rout} \subseteq \diskY{q}{r_0+C\Gamma}.
    \]
    Thus every boundary grid point of $\Popt$ lies in the annulus.  The
    inner disk of this annulus is strictly contained in $\Popt$, so no
    supporting line of a boundary link of $\Popt$ meets its interior.
    Consequently the boundary links of $\Popt$ pass the same visibility
    and tangent tests used in the annulus graph.
\end{proof}

\subsubsection{Putting everything together}

\begin{theorem}
    \thmlab{dp-leftmost}
    The minimum-perimeter convex grid polygon enclosing~$k$ grid points can be computed in~$O(k^{29/18+o(1)})$ time.
\end{theorem}

\begin{proof}
    Let $\Gamma=C_\Gamma k^{1/6}$, where $C_\Gamma$ is chosen large enough
    for the disk-sandwich and radius estimates from \thmref{improved-gap}
    and the proof of \lemref{centered-dp}.  Let $r_0=\sqrt{k/\pi}$, and
    let $\lambda_0=c_\ell k^{1/4}$, where $c_\ell>0$ is a certified
    lower-bound constant from \lemref{leftmost-edge}.  Also set
    $H=C_H k^{5/18}$, where $C_H$ is a sufficiently large absolute
    constant from \lemref{leftmost_midpoint}.

    The algorithm is straightforward.  We build a small set of candidate centers, run the centered \DP of \lemref{centered-dp} once for each of them, and return the best feasible polygon.  The only point to prove is that one of these candidates is always within $O(\Gamma)$ of the correct center.

    Define
    \[
        T = \Set{ j\lambda_0/2 }{ j\in\ZZ,\, |j\lambda_0/2|\leq H+\lambda_0 }.
    \]
    Then
    \[
        \cardin{T} = O(H/\lambda_0+1) = O(k^{5/18-1/4}) = O(k^{1/36}).
    \]
    For each $t\in T$, run the centered \DP with candidate center
    \[
        q_t=(t,r_0).
    \]

    Fix an optimal polygon~$\Popt$, and let $c$ be a \defrefY{epicenter}{epicenter}.  Work in the translated coordinates of \lemref{leftmost_midpoint}, so that $c$ is the origin.  The grid is now the translate $L=\ZZ^2-c$.  Let $e_\ell$ be the leftmost edge, let $\lambda=\lenX{e_\ell}$, and let $p_\ell=(x_\ell,y_\ell)$ be its midpoint.  By \lemref{leftmost-edge} and \lemref{leftmost_midpoint},
    \[
        \lambda\geq\lambda_0, \qquad |y_\ell|\leq H, \qquad -x_\ell=r_0+O(\Gamma).
    \]
    Since $e_\ell$ is vertical, the $y$-coordinates of its points fill the interval
    \[
        I=[y_\ell-\lambda/2,\, y_\ell+\lambda/2].
    \]
    This interval has length at least $\lambda_0$ and midpoint in $[-H, H]$.  Since the points of $T$ are spaced by $\lambda_0/2$ and cover $[-H-\lambda_0,H+\lambda_0]$, there exists $t\in T\cap I$.

    Because $e_\ell$ is a vertical segment in the translated lattice~$L$,
    it contains a boundary lattice point $s=(x_\ell,s_y)$ with
    $|s_y-t|\leq 1/2$.  Write $s=z-c$ for some $z\in\ZZ^2$.  Then
    translating by $-s$ sends $L$ to $\ZZ^2-z=\ZZ^2$, and a
    counterclockwise $90^\circ$ rotation preserves $\ZZ^2$.  Translate
    $\Popt$ by $-s$ and apply this rotation. Denote the resulting polygon
    by $P_s$.  After the translation, $\Popt$ lies in the half-plane
    $x\geq 0$ because $s$ lies on the left supporting line, so after the
    rotation $P_s$ lies in the half-plane $y\geq 0$.  Thus the origin is a
    bottom support grid point of $P_s$.

    The \defrefY{epicenter}{epicenter} of $P_s$ is obtained from $c=0$ by the same translation and rotation, namely
    \[
        c_s=(s_y,-x_\ell).
    \]
    This is why we enumerate $y$-coordinates on the leftmost edge: after the rigid motion, that $y$-coordinate becomes the first coordinate of the center.  Hence
    \[
        |(c_s)_x-t|\leq \frac12, \qquad |(c_s)_y-r_0|=|-x_\ell-r_0|=O(\Gamma),
    \]
    so $\lenX{q_t-c_s}=O(\Gamma)$.  By \lemref{centered-dp}, the \DP run
    with center $q_t$ contains the boundary path of~$P_s$, so its optimum
    value is at most $\perimX{P_s}=\lOpt$.  Conversely, every completed
    \DP solution yields, by the usual convex-hull and trimming step, a
    feasible $k$-point grid polygon.  Its perimeter is therefore at least
    $\lOpt$.  Hence the best polygon returned by this run is optimal.

    Since each run costs $O(k^{19/12+o(1)})$, trying all candidates in $T$ takes total time
    \[
        \cardin{T}\cdot O(k^{19/12+o(1)}) = O(k^{1/36})\cdot O(k^{19/12+o(1)}) = O(k^{29/18+o(1)}).
    \]
\end{proof}

\section{Implementation and experiments}
\seclab{experiments}

We implemented the \DP algorithm described here in both Rust and C++.
The source code is available at
\url{https://github.com/sarielhp/rust_k_perimeter}.  A webpage with an
interface to explore the optimal polygons computed is available at
\url{https://sarielhp.org/r/26/k_perimeter}.  We used Agentic AI to
optimize and write much of the code (specifically, Gemini and
Claude-code). This enabled us to implement various speedups that in the
past would not have been implemented because they tend to be quite
tedious. In particular, the implementation tries to precompute any useful
information that is used during the \DP, so that the \DP running time is
minimized. For example, the implementation precomputes the primitive
directions, visibility tests, and grid-point counts needed by
transitions, so that the main \DP loop performs only constant-time work
per processed edge.

We believe that some of the ideas we used here should be useful for other \DP problems. In particular, we were able to solve in practice instances where the \DP memory footprint exceeds the memory available on the computer we used.

We did not implement the more complicated parts of the theoretical
algorithm --- \secref{monotonicity} and \secref{rings_2} --- since the
practical implementation was already fast without these optimizations.
This leaves a gap between the strongest theoretical running time proved
above and the simpler algorithm used in the experiments.

The result was that we were able to solve very large instances (i.e., $k=2.5\cdot 10^6$) in less than two hours, using a reasonable desktop computer, see \tblref{r_t}. This becomes more impressive when one realizes that the basic algorithm cannot achieve this even for $k=10,000$. One can surely get further improvement with more work, but the results seem impressive enough that this seems like a natural stopping point.

\paragraph*{On the limited benefit of experiments.}

We tried to use our experiments to figure out what the true values of the various parameters our analysis used (such as number of vertices, width of annulus, etc) are. The problem is that for $k$ in the millions, the quantity $k^{1/6}$ is roughly $10$. At this point, trying to divine asymptotic bounds from experimental results seems pointless. To do that, and get results that can be trusted, one would probably need to run experiments with $k$ in the range $10^{9} \ldots 10^{28}$, which is well outside the current ability of our program (and even then, one might need to use even larger values of $k$). We thus shared only a few graphs summarizing our experimental results.

\subsection{Additional heuristics}

We implemented several heuristics that significantly improved the
performance in practice.

\begin{compactenumA}[wide, labelindent=0pt, leftmargin=10pt]
    \item \textsf{Too many grid points would be covered.} For each point
    $v \in \VV$, the algorithm precomputes the ``shortest'' path from $v$
    to $o'$. Specifically, this is the polygonal path between $v$ to $o'$
    covering the minimum number of grid points not yet covered --- this
    convex-chain is formed by the shortest path avoiding the forbidden
    points in the middle.
    This is a convex chain $C_v$, and the algorithm computes how many
    grid points this region adds (which must be included in any solution)
    if we attach it to a solution ending at $v$. Let $g_{\min,v}$ denote
    this number.
    Thus, when a new configuration ending at $v$ is
    considered, if the current number of grid points it covers, plus
    $g_{\min,v}$, exceeds $k$, the algorithm rejects this configuration
    (i.e., it does not queue it for further inspection).

    \item \textsf{Too few grid points would be covered.}  Similarly,
    given a new configuration $(v,\nG)$, and the edge $u \rightarrow v $ that
    arrived to it, one can compute an upper bound on the additional grid
    points this solution can cover.  Formally, all the additional points
    that might be covered by a solution using this edge must be to the
    left (or on) the vector $u \rightarrow v$, and to the left of the vector
    $\origin \rightarrow v$. Let $C_{u\rightarrow v}$ denote this cone.  The algorithm
    computes
    $\nG_{\max,u\to v} = \cardin{ ( B \cup \VV ) \cap C_{u \rightarrow v }}$, see
    \defref{forbidden}. To compute this number quickly, the
    implementation uses a $kd$-tree on the point set to perform this
    range counting query. If the new configuration under consideration by
    the edge $u \rightarrow v $ covers $\nG$ grid points currently, and
    $\nG + \nG_{\max,u\to v} < k$, then this configuration is rejected.

    \item \textsf{Topological ordering.} The exact order used by the
    queue seems to matter quite a bit for performance. We tried several
    different strategies, but the following seems to work the best ---
    the algorithm performs a topological ordering of the vertices of the
    graph, and numbers these vertices accordingly
    $1, \ldots, \cardin{\VV}$. The algorithm first handles the
    configurations with smaller topological index in the queue. Among
    configurations with the same topological index, it handles the
    configurations in increasing order of the number of grid points the
    partial solution covers.
\end{compactenumA}

\medskip%
\noindent%
The implementation maintains an array \ArrConf of configurations computed. The algorithm also maintains a hash table for looking up computed configurations. The configurations are stored in the hash table using a \emphw{key}, which is a pair $(i,\nG)$, where $i$ is the topological index of the vertex, and $\nG$ is the number of grid points covered so far. We store within each such configuration (i.e., in the array $\ArrConf$) the minimum perimeter partial solution computed so far ending at this configuration, the edge of the visibility graph used to arrive at this solution, and the index of the previous configuration that led to this configuration.  Thus, the algorithm can recover the optimal polygon at the end of the \DP execution.

\begin{compactenumA}[resume,wide, labelindent=0pt, leftmargin=10pt]

    \item \textsf{Memory mapping the configuration space.} The problem is that both $\ArrConf$ and $\Hash$ can become humongous as far as memory consumption (e.g., $k=10^6$). Fortunately, the access pattern to the array \ArrConf is almost linear --- only configurations that are close to the currently handled vertex (say with topological index $i$) may be updated. All entries in the array that have topological index $<i$ are relevant only for the final stage of recovering the solution. Thus, we memory-mapped the array $\ArrConf$ into a file using an SSD \cite{aa-ostep-23}. This enables us to let \ArrConf grow well past the memory available in the system, without dramatic slowdown in computation.
\end{compactenumA}

\begin{remark}
    There are standard techniques to improve the theoretical space
    complexity, trading space for time (used, for example, in reducing
    the edit-distance \DP space to linear). In our case, one can chop the
    ring into several pieces, and only store the \DP values at the
    beginning region $R_i$ of each piece which is contained in a rectangle with width $\emax$ and
    height $O(\WD)$ (as a \DP transition cannot completely cross $R_i$). When
    solutions recover the solution in the end, rerun the DP in each
    piece. We have not explored this idea further, but it might be the
    next natural step in trying to scale our implementation to larger
    inputs.
 \end{remark}

\begin{compactenumA}[resume,wide, labelindent=0pt, leftmargin=10pt]
    \item \textsf{Cleaning the hash table.}  Since the hash table is usable only for future configurations that have a higher topological index than the one being handled by the \DP loop, the algorithm, every once in a while, cleans the hash table, deleting all the configurations that are in the past and are no longer usable. The implementation uses a simple doubling threshold to decide when to perform the cleanup. It first sets some (relatively small) arbitrary threshold that depends on $k$ for the cleanup. Whenever the hash table $\Hash$ stores more than this threshold number of entries in the table, the implementation stops and performs the cleanup. The threshold is updated to be (say) twice the number of entries left after the cleanup. This has the effect that the memory used by the hash-table is effectively the width of the \DP. Since this problem has relatively low width, this results in a significant reduction in memory usage and increase in speed.

\end{compactenumA}

\subsubsection{Tightening the parameters}

So far, the heuristics described do not impact correctness. Next, we describe a few additional heuristics that were based on experimental results to fine-tune and speed up the execution. We verified our results up to $k=10,000$ without these heuristics, so we are quite confident the results of our program are correct, but it is quite conceivable that these guesses are wrong. For very large values of $k$, these guesses would fail and lead to an incorrect solution.

\begin{compactenumA}[resume]

    \item \textsf{Width for marking good/bad points.} %
    The algorithm computes the circle used in \lemref{tight-upper}, and
    places its bottom point at the origin. Then it marks all the points
    in distance $\leq 2 + k^{1/4}/4 $ from this circle as being good. By
    default, all other points are bad. The algorithm also shifts this
    circle a bit (roughly $\Theta(k^{1/4})$) to get enough margin of good
    points around the origin.

    \item \textsf{Turning angle.} As the optimal solution converges to a disk as $k$ increases, the optimal polygon has more edges, and the turning angles become smaller. As such, our implementation considers only edges with turning angles within the range $[0,2 \pi / k^{1/3}] $. This range is roughly what one would expect from a moderately turning polygon, as the maximum number of vertices of the optimal solution polygon is $O(k^{1/3})$. If we expect the turning angle to be roughly uniform at the vertices, this is the right quantity. Our experiments indicate that this upper bound is close to but is well above the maximum turning angle in the optimal solution.

    \item \textsf{Longest primitive vector.} %
    Similarly, the program enforces an upper bound of $\ceil{k^{1/3}}+1$ on the $x$ and $y$ extent of all the primitive vectors being considered. The set of primitive vectors is computed using Farey's sequence \cite{hw-itn-08}, which enables us to compute the primitive vector in linear time in their number.

\end{compactenumA}

\subsection{Experimental results}

We ran extensive experiments with our implementation, using a Linux machine with 64GB of memory and AMD Ryzen 7 5700G (a four-year-old CPU). Disappointingly, it is (roughly) only half the speed of the fastest CPU currently available in single-thread performance. The running times of our implementation are shown in \tblref{r_t}. Note that the running times seem a bit strange, as the interaction of the virtual memory together with the file-mapped mechanism has a somewhat unpredictable effect on the running time.

\begin{table}[p]
    % flatex input: [figs/rt_table.tex]
  \centering
  \small
  \begin{tabular}[t]{rr}
    \toprule
    $k$ Value & Time (s) \\
    \midrule
    10 & 0.03 \\
    20 & 0.02 \\
    50 & 0.02 \\
    100 & 0.03 \\
    500 & 0.05 \\
    1,000 & 0.09 \\
    5,000 & 0.83 \\
    10,000 & 3.61 \\
    \bottomrule
  \end{tabular}
  \hfill
  \begin{tabular}[t]{rr}
    \toprule
    $k$ Value & Time (s) \\
    \midrule
    20,000 & 10.14 \\
    30,000 & 18.98 \\
    40,000 & 42.61 \\
    50,000 & 46.30 \\
    60,000 & 73.37 \\
    70,000 & 99.33 \\
    80,000 & 117.15 \\
    90,000 & 130.22 \\
    \bottomrule
  \end{tabular}
  \hfill
  \begin{tabular}[t]{rr}
    \toprule
    $k$ Value & Time (s) \\
    \midrule
    100,000 & 131.60 \\
    150,000 & 385.86 \\
    200,000 & 556.94 \\
    250,000 & 947.47 \\
    300,000 & 1,214.20 \\
    350,000 & 1,525.66 \\
    1,500,000 & 4,716.69 \\
    2,500,000 & 6,827.01 \\
    \bottomrule
  \end{tabular}

    \caption{Performance summary split into three independent parallel
       sections.}
    \tbllab{r_t}
\end{table}

\begin{figure}[p]
    \centering \includegraphics[width=0.8\linewidth]{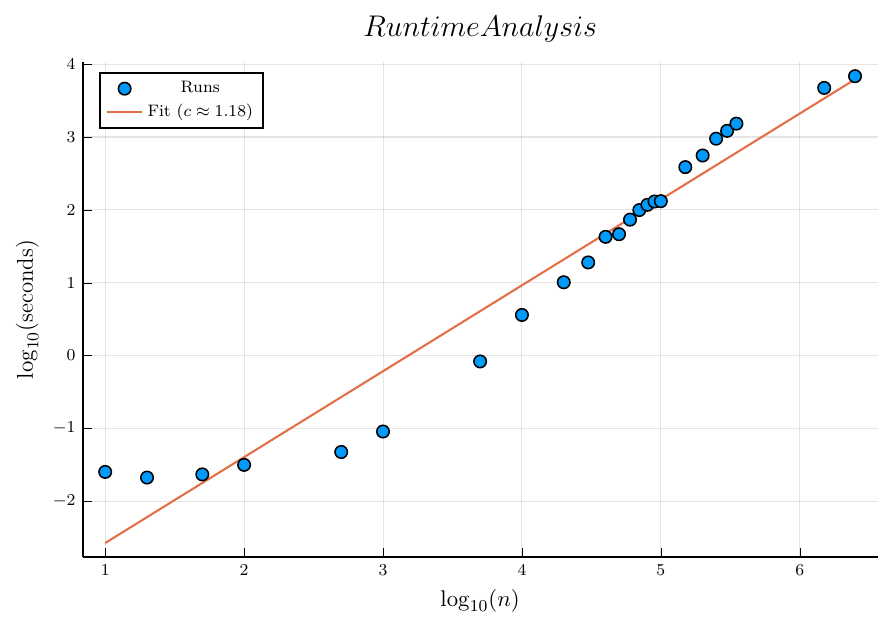}
    \caption{Fitting a line in the $\log/\log$ space.}
    \figlab{r_t_fit}
\end{figure}

We also performed a statistical analysis of various parameters of the optimal solutions and tried to figure out how these quantities grow as a function of $k$. This was done by doing the best linear fit in the $\log/\log$ space of the relevant parameters.

\figref{r_t_fit} shows the running time fitting. It suggests the current algorithm running time is $O(k^{1.42})$.  Such results, being based on experimental evidence, should be taken with appropriate skepticism. Nevertheless, we believe the results are interesting and present them here.

\paragraph{Number of configurations computed.}
Maybe a better proxy for the running time is the number of distinct configurations computed. Note that this does not quite bound the running time, as each configuration involves additional work in the \DP. The graph of \figref{configs} suggests the number of configurations grows at the rate of $O(k^{5/4})$.

\begin{figure}[ht]
    \centering%
    \includegraphics[width=0.9\linewidth]{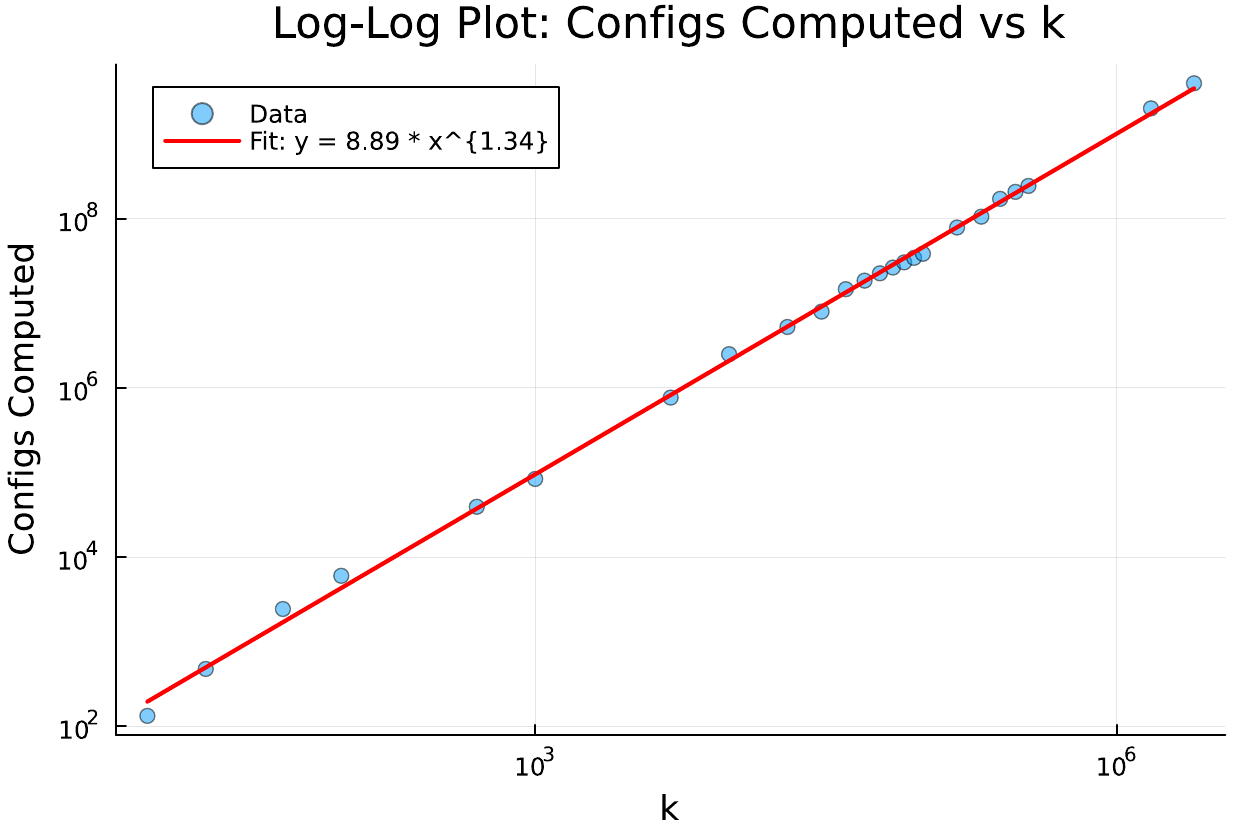}
    \caption{\# Configurations computed as a function of $k$.}
    \figlab{configs}
\end{figure}

\section{Conclusions}

In this paper, we presented fast algorithms, both in theory and practice,
for computing the smallest perimeter polygon containing $k$ grid
points. Starting with a baseline of $O(k^3)$ dynamic programming
algorithm for this problem, we were able to improve the running time to
(roughly) $O(k^{1.62})$ (we believe further improvement should be
possible). Furthermore, starting with an implementation that struggled
with solving instances for $k$ in the low few thousands, our current
implementations solved the problem successfully for $k=2.5\cdot 10^6$. We
benefited from going back and forth between the implementation and the
theory---the theory suggesting how to improve the implementation, and the
applied results suggesting what structural properties the optimal
solution has that one should prove, etc.

\subsection{AI Acknowledgement}
For transparency, we describe the role that AI tools played in this
work. Essentially (almost) all of the code implementation was
AI-generated: we used agentic coding assistants (in particular Gemini and
Claude Code) to write, optimize, and debug the Rust and C++ code,
including several of the more tedious speedups that we likely would not
have implemented by hand --- we emphasize however that the ideas described
in the paper both on the theory and the applied side were the authors
ideas. On the mathematical side, the results and their proofs are the
authors', but we relied heavily on large language models (ChatGPT,
Gemini, and Claude, among others) to proofread the proofs, expand
arguments that were too terse, and fill in routine gaps in the
write-up. All AI-assisted proofs and text were subsequently read and
verified by the authors. This process was occasionally substantive: for
example, it helped surface a genuine gap in the monotone-queue speedup of
\secref{monotonicity} (\lemref{eligible-envelope}), which we were then
able to patch. Overall, we found AI to be a remarkably useful
collaborator throughout the project, both as a writing and pedagogical
aid and as an accelerator for the implementation and
experiments. Responsibility for any remaining errors rests, as always,
with the authors.

On the other hand, the AI programming tools repeatedly tried to simplify
our implementation by removing certain optimizations as unnecessary
(specifically, the memory mapping).

\printbibliography

\appendix

\section{\texorpdfstring{More examples of optimal $k$-polygons}{More examples of optimal k-polygons}}

\begin{figure}[h]
    \foreach \i in {11,...,15}{%
       \hfill \includegraphics[page=\i, width=0.18\linewidth]{figs/polygons}%
    }%
    \par\bigskip
    \foreach \i in {16,...,20}{%
       \hfill \includegraphics[page=\i, width=0.18\linewidth]{figs/polygons}%
    }%
    \par\bigskip
    \foreach \i in {21,...,25}{%
       \hfill \includegraphics[page=\i, width=0.18\linewidth]{figs/polygons}%
    }%
    \par\bigskip
    \foreach \i in {26,...,30}{%
       \hfill \includegraphics[page=\i, width=0.18\linewidth]{figs/polygons}%
    }%
    \par\bigskip
    \foreach \i in {31,...,35}{%
       \hfill \includegraphics[page=\i, width=0.18\linewidth]{figs/polygons}%
    }%
    \caption{Optimal polygons for $k=13, \ldots, 37$}
    \figlab{rest_1}
\end{figure}

\begin{figure}
    \par\bigskip
    \foreach \i in {36,...,40}{%
       \hfill \includegraphics[page=\i, width=0.18\linewidth]{figs/polygons}%
    }%
    \caption{Optimal polygons for $k=38, \ldots, 42$}
    \figlab{rest_1b}
\end{figure}

\begin{figure}
    \foreach \i in {48,53,58,63,68}{%
       \hfill \includegraphics[page=\i, width=0.18\linewidth]{figs/polygons}%
    }%
    \par\bigskip
    \foreach \i in {73,78,83,88,93}{%
       \hfill \includegraphics[page=\i, width=0.18\linewidth]{figs/polygons}%
    }%
    \caption{Optimal polygons for $k=50, \ldots, 95$}
    \figlab{rest_2}
\end{figure}
\begin{figure}
    \foreach \i in {1,...,5}{%
       \hfill \includegraphics[page=\i, width=0.18\linewidth]{figs/polygons_100_1000}%
    }%
    \par\bigskip
    \foreach \i in {6,...,10}{%
       \hfill \includegraphics[page=\i, width=0.18\linewidth]{figs/polygons_100_1000}%
    }%
    \caption{Optimal polygons for $k=100, \ldots, 1000$}
    \figlab{rest_3}
\end{figure}

\begin{figure}
    \foreach \i in {1,...,3}{%
       \hfill \includegraphics[page=\i, width=0.32\linewidth]{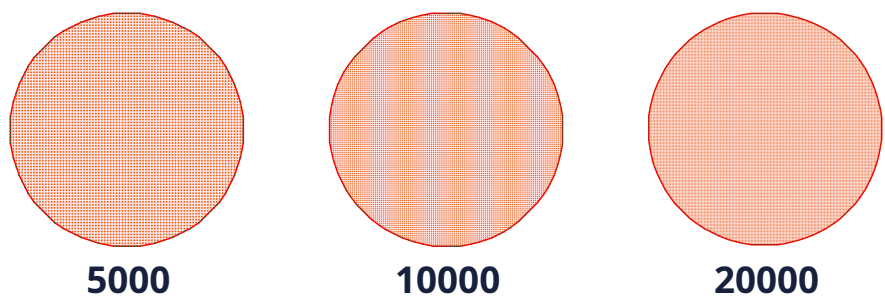}%
    }%
    \caption{Optimal polygons for $k=5000$, $k=10,000$ and $k=20,000$.  }
    \figlab{rest_4}
\end{figure}

\end{document}